%% file: CFG_arxiv_v2.tex
\documentclass[11pt]{article} 
\usepackage{fullpage}
\usepackage[font=footnotesize]{caption}
\usepackage{comment}
\usepackage{amsmath}
\usepackage{amsthm}
\usepackage{amssymb}
\usepackage{url}
\usepackage[final]{showkeys}
\usepackage{cleveref}
\usepackage{tikz}
\usepackage{lineno}

\newenvironment{reminder}[1]{\smallskip

\noindent {\bf Reminder of #1 }\em}{\smallskip}

\newtheorem{theorem}{Theorem}
\newtheorem{lemma}{Lemma}

\newtheorem{claim}{Claim}

\newtheorem{definition}{Definition}

{\makeatletter
 \gdef\xxxmark{%
   \expandafter\ifx\csname @mpargs\endcsname\relax 
     \expandafter\ifx\csname @captype\endcsname\relax 
       \marginpar{xxx}
     \else
       xxx 
     \fi
   \else
     xxx 
   \fi}
 \gdef\xxx{\@ifnextchar[\xxx@lab\xxx@nolab}
 \long\gdef\xxx@lab[#1]#2{{\bf [\xxxmark #2 ---{\sc #1}]}}
 \long\gdef\xxx@nolab#1{{\bf [\xxxmark #1]}}
}

\def \eps {\varepsilon}

\begin{document}

\def \isnotin {\nsubseteq}

\def \eps {\varepsilon}
\title{If the Current Clique Algorithms are Optimal, \\ so is Valiant's Parser}
\author{
	Amir Abboud\\ Stanford University \\ \texttt{\small abboud@cs.stanford.edu}
	\and Arturs Backurs\\ MIT \\ \texttt{ \small backurs@mit.edu}
	\and Virginia Vassilevska Williams\\ Stanford University \\ \texttt{\small virgi@cs.stanford.edu}
}
\date{}
\begin{titlepage}
\clearpage\maketitle
\thispagestyle{empty}

\abstract{
The CFG recognition problem is: given a context-free grammar $\mathcal{G}$ and a string $w$ of length $n$, decide if $w$ can be obtained from $\mathcal{G}$.
This is the most basic parsing question and is a core computer science problem.
Valiant's parser from 1975 solves the problem in $O(n^{\omega})$ time, where $\omega<2.373$ is the matrix multiplication exponent.
Dozens of parsing algorithms have been proposed over the years, yet Valiant's upper bound remains unbeaten.
The best \emph{combinatorial} algorithms have mildly subcubic $O(n^3/\log^3{n})$ complexity.



Lee (JACM'01) provided evidence that fast matrix multiplication is needed for CFG parsing, and that very efficient and practical algorithms might be hard or even impossible to obtain. Lee showed
that any algorithm for a more general parsing problem with running time $O(|\mathcal{G}|\cdot n^{3-\eps})$ can be converted into a surprising subcubic algorithm for Boolean Matrix Multiplication. Unfortunately, Lee's hardness result required that the grammar size be $|\mathcal{G}|=\Omega(n^6)$.
Nothing was known for the more relevant case of constant size grammars.

In this work, 
we prove that \emph{any} improvement on Valiant's algorithm, even for constant size grammars, either in terms of runtime or by avoiding the inefficiencies of fast matrix multiplication, would imply a breakthrough algorithm for the $k$-Clique problem: given a graph on $n$ nodes, decide if there are $k$ that form a clique.

Besides classifying the complexity of a fundamental problem, our reduction has led us to similar lower bounds for more modern and well-studied cubic time problems for which faster algorithms are highly desirable in practice: \emph{RNA Folding}, a central problem in computational biology, and \emph{Dyck Language Edit Distance}, answering an open question of Saha (FOCS'14).
}
\end{titlepage}

\section{Introduction}
\input{CFGintro}

\section{Clique to CFG Recognition}
\label{sec:cfg}
\input{proof}

\section{Clique to RNA folding}
\label{sec:RNA}
\input{rna}

\section{Clique to Dyck Edit Distance}
\label{sec:dyck}
\input{dyck_construction}

\input{dyck_part1}

\medskip
\paragraph{Acknowledgements.}
We would like to thank Piotr Indyk for a discussion that led to this work, and
Roy Frostig for introducing us to many modern works on CFG parsing.
We also thank Alex Andoni, Ryan Williams, and the anonymous reviewers for comments.
A.B. was supported by NSF and Simons Foundation.
 A.A. and V.V.W. were supported by a Stanford School of Engineering Hoover Fellowship, NSF Grant CCF-1417238, NSF Grant CCF-1514339, and BSF Grant BSF:2012338.

\bibliographystyle{alpha}
\bibliography{ref}

\end{document}

%% file: CFGintro.tex
Context-free grammars (CFG) and languages (CFL), introduced by Chomsky in 1956 \cite{chomsky}, play a fundamental role in computability theory \cite{sipserbook}, formal language theory \cite{HopcroftAutomata}, programming languages \cite{Aho86}, natural language processing \cite{Jurafsky2000}, and computer science in general with applications in diverse areas such as computational biology \cite{biobook} and databases \cite{Korn13}.
They are essentially a sweet spot between very expressive languages (like natural languages) that computers cannot parse well, and the more restrictive languages (like regular languages) that even a DFA can parse.
%

In this paper, we will be concerned with the following very basic definitions.
A CFG $\mathcal{G}$ in \emph{Chomsky Normal Form} over a set of terminals (i.e. alphabet) $\Sigma$ consists of a set of nonterminals $\mathcal{T}$, including a specified starting symbol $\mathbf{S} \in \mathcal{T}$, and a set of productions (or derivation rules) of the form $\mathbf{A \to B \ C}$ or $\mathbf{A} \to \sigma$ for some $\mathbf{A,B,C} \in \mathcal{T}$ and $\sigma \in \Sigma$.
Each CFG $\mathcal{G}$ defines a CFL $\mathcal{L(G)}$ of strings in $\Sigma^*$ that can be obtained by starting with $\mathbf{S}$ and recursively applying arbitrary derivation rules from the grammar.
The \emph{CFG recognition} problem is: given a CFG $\mathcal{G}$ and a string $w \in \Sigma^*$ determine if $w$ can be obtained from $\mathcal{G}$ (i.e. whether $w \in \mathcal{L(G)}$).
The problem is of most fundamental and practical interest when we restrict $\mathcal{G}$ to be of fixed size and let the length of the string $n= |w|$ to be arbitrary.

The main question we will address in this work is: \emph{what is the time complexity of the CFG recognition problem?}

Besides the clear theoretical importance of this question, the practical motivation is overwhelming.
CFG recognition is closely related to the \emph{parsing} problem in which we also want to output a possible derivation sequence of the string from the grammar (if $w \in \mathcal{L(G)}$). Parsing is essential: this is how computers understand our programs, scripts, and algorithms.
Any algorithm for parsing solves the recognition problem as well, and Ruzzo~\cite{Ruzzo} showed that CFG recognition is at least as hard as parsing, at least up to logarithmic factors, making the two problems roughly equivalent.


Not surprisingly, the critical nature of CFG recognition has led to the development of a long list of clever algorithms for it, including classical works \cite{Valiant,earley,cocke,younger,kasami,knuthlr,deremer,insideoutside,clrparser}, and the search for practical parsing algorithms, that work well for varied applications, is far from over \cite{PK09,RSCJ10,Soc13,CSC13}.
For example, the canonical CYK algorithm from the 1960's \cite{cocke,kasami,younger} constructs a dynamic programming table $D$ of size $n\times n$ such that cell $D(i,j)$ contains the list of all nonterminals that can produce the substring of $w$ from position $i$ to position $j$.
The table can be computed with linear time per entry, by enumerating all derivation rules $\mathbf{A \to B \ C}$ and checking whether for some $i\leq k\leq j$, $D(i,k)$ contains $\mathbf{B}$ and $D(k+1,j)$ contains $\mathbf{C}$ (and if so, add $\mathbf{A}$ to $D(i,j)$).
This gives an upper bound of $O(n^3)$ for the problem.
Another famous algorithm is Earley's from 1970 \cite{earley} which proceeds by a top-down dynamic programming approach and could perform much faster when the grammar has certain properties.
Variants of Earley's algorithm were shown to run in mildly subcubic $O(n^3/\log^2 n)$ time \cite{Graham80,Rytter85}. \footnote{As typical, we distinguish between ``truly subcubic'' runtimes, $O(n^{3-\eps})$ for constant $\eps>0$, and ``mildly subcubic'' for all other subcubic runtimes.} 

In 1975 a big theoretical breakthrough was achieved by Valiant~\cite{Valiant} who designed a sophisticated recursive algorithm that is able to utilize many fast boolean matrix multiplications to speed up the computation of the dynamic programming table from the CYK algorithm.
The time complexity of the CFG problem decreased to $O(g^2 n^{\omega})$, where $\omega<2.373$ is the matrix multiplication exponent~\cite{v12,legallmult} and $g$ is the size of the grammar; because in most applications $g=O(1)$, Valiant's runtime is often cited as $O(n^\omega)$.
In 1995 Rytter~\cite{Rytter95} described this algorithm as ``\emph{probably the most interesting algorithm related to formal languages}" and it is hard to argue with this quote even today, 40 years after Valiant's result.
Follow-up works 
proposed simplifications of the algorithm \cite{Rytter95}, generalized it to stochastic CFG parsing \cite{BS07}, and applied it to other problems~\cite{akutsu,ZTZ10}.

Despite its vast academic impact, Valiant's algorithm has enjoyed little success in practice.
The theoretically fastest matrix multiplication algorithms are not currently practical, and Valiant's algorithm can often be outperformed by \emph{``combinatorial''} methods in practice, even if the most practical truly subcubic fast matrix multiplication algorithm (Strassen's~\cite{strassen}) is used.
Theoretically, the fastest combinatorial algorithms for Boolean Matrix Multiplication (BMM) run in time $O(n^3/\log^4{n})$ \cite{Chan15,huachengbmm}.
To date, no combinatorial algorithm for BMM or CFG recognition with truly subcubic running time is known.

In the absence of efficient algorithms and the lack of techniques for proving
superlinear unconditional lower bounds for any natural problem, researchers have turned to conditional lower bounds for CFG recognition and parsing.
Since the late 1970's, Harrison and Havel \cite{HH74} observed that any algorithm for the problem would imply an algorithm that can verify a Boolean matrix multiplication of two $\sqrt{n} \times \sqrt{n}$ matrices.
This reduction shows that a combinatorial $O(n^{1.5-\eps})$ recognition algorithm would imply a breakthrough subcubic algorithm for BMM.
Ruzzo \cite{Ruzzo} showed that a parsing algorithm that says whether each prefix of the input string is in the language, could even compute the BMM of two $\sqrt{n}\times \sqrt{n}$ matrices.
Even when only considering combinatorial algorithms, these $\Omega(n^{1.5})$ lower bounds left a large gap compared to the cubic upper bound.
A big step towards \emph{tight} lower bounds was in the work of Satta \cite{Satta} on parsing Tree Adjoining Grammars, which was later adapted by Lee \cite{Lee} to prove her famous conditional lower bound for CFG parsing.
Lee proved that BMM of two $n \times n$ matrices can be reduced to ``parsing" a string of length $O(n^{1/3})$ with respect to a CFG of size $\Theta(n^2)$, where the parser is required to say for each nonterminal $T$ and substring $w[i:j]$ whether $T$ can derive $w[i:j]$ in a valid derivation of $w$ from the grammar.
This reduction proves that such parsers cannot be combinatorial and run in $O(gn^{3-\eps})$ time without implying a breakthrough in BMM algorithms. 

Lee's result is important, however suffers from significant limitations which have been pointed out by many researchers (e.g. \cite{Ruzzo,Lee,SahaDyckApprox,SahaArxiv}). We describe a few of these below. Despite the limitations, however, the only progress after Lee's result is a recent observation by Saha \cite{Saha15} that one can replace BMM in Lee's proof with APSP by augmenting the production rules with probabilities, thus showing an APSP-based lower bound for \emph{Stochastic CFG} parsing.
Because Saha uses Lee's construction, her lower bound suffers from exactly the same limitations.



%

The first (and most major) limitation of Lee's lower bound is that it is irrelevant unless the size of the grammar is much larger than the string, in particular it is cubic only when $g=\Omega(n^6)$.
A CFG whose description needs to grow with the input string does not really define a CFL, and as Lee points out, this case can be unrealistic in many applications.
In programming languages, for instance, the grammar size is much smaller than the programs one is interested in, and in fact the grammar can be hardcoded into the parser.
A parsing algorithm that runs in time $O(g^{3} n)$, which is \emph{not} ruled out by Lee's result, could be much more appealing than one that runs in  $O(gn^{2.5})$ time.

The second limitation of both Lee's and Ruzzo's lower bounds is the quite demanding requirement from the parser to provide extra information besides returning some parse tree.
These lower bounds do not hold for recognizers nor any parser with minimal but meaningful output.

Theoretically, it is arguably more fundamental to ask: what is the time complexity of CFG recognition and parsing that can be obtained by \emph{any} algorithm, not necessarily combinatorial?
Lee's result cannot give a meaningful answer to this question.
To get a new upper bound for BMM via Lee's reduction, one needs a parser that runs in near-linear time.
Lee's result does not rule out, say, an $O(n^{1.11})$ time parser that uses fast matrix multiplication; such a parser would be an amazing result.


Our first observation is that to answer these questions and understand the complexity of CFG recognition, we may need to find a problem other than BMM to reduce from.
Despite the apparent similarity in complexities of both problems - we expect both to be cubic for combinatorial algorithms and $O(n^{\omega})$ for unrestricted ones - there is a big gap in complexities because of the input size.
When the grammar is fixed, a reduction cannot encode any information in the grammar and can only use the $n$ letters of the string, i.e. $O(n)$ bits of information, while an instance of BMM requires $\Theta(n^2)$ bits to specify. Thus, at least with respect to reductions that produce a single instance of CFG recognition, we do not expect BMM to imply a higher than $\Omega(n^{1.5})$ lower bound for parsing a fixed size grammar.

\paragraph{Main Result}
In this paper we present a tight reduction from the $k$-Clique problem to the recognition of a fixed CFG and prove a new lower bound for CFG recognition that overcomes all the above limitations of the previously known lower bounds.
Unless a breakthrough $k$-Clique algorithm exists, our lower bound completely matches Valiant's 40-year-old upper bound for unrestricted algorithms and completely matches CYK and Earley's for combinatorial algorithms, thus resolving the complexity of CFG recognition even on fixed size grammars.

Before formally stating our results, let us give some background on $k$-Clique.
This fundamental graph problem asks whether a given undirected unweighted graph on $n$ nodes and $O(n^2)$ edges contains a clique on $k$ nodes. 
This is the parameterized version of the famously NP-hard Max-Clique (or equivalently, Max-Independent-Set) \cite{Karp72}.
$k$-Clique is amongst the most well-studied problems in theoretical computer science, and it is the canonical intractable (W[1]-complete) problem in parameterized complexity.

A naive algorithm solves $k$-Clique in $O(n^k)$ time.
By a reduction from 1985 to BMM on matrices of size $n^{k/3} \times n^{k/3}$ it can be solved with fast matrix multiplication in $O(n^{\omega k/3})$ time \cite{NP85} whenever $k$ is divisible by $3$ (otherwise, more ideas are needed \cite{EG04}). No better algorithms are known, and researchers have wondered if improvements are possible~\cite{woeginger,babai}.
As is the case for BMM, obtaining faster than trivial \emph{combinatorial} algorithms, by more than polylogarithmic factors, for $k$-Clique is a longstanding open question.
The fastest combinatorial algorithm runs in $O(n^k/\log^k{n})$ time \cite{VClique}.

Let $0 \leq F \leq \omega$ and $0 \leq C \leq 3$ be the smallest numbers such that $3k$-Clique can be solved combinatorially in $O(n^{Ck})$ time and in $O(n^{Fk})$ time by any algorithm, for any (large enough) constant $k\geq 1$.
A conjecture in graph algorithms and parameterized complexity is that $C=3$ and $F = \omega$.
It is known that an algorithm refuting this conjecture immediately implies a faster exact algorithm for MAX-CUT \cite{williams2005new,woeginger2008open}. 
Note that even a linear time algorithm for BMM ($\omega = 2$) would not prove that $F < 2$.
A well known result by Chen et al. \cite{ChenCFHJKX05,ChenHKX06} shows that $F>0$ under the Exponential Time Hypothesis. A plausible conjecture about the parameterized complexity of Subset-Sum implies that $F \geq 1.5$ \cite{ALW14}. 
There are many other negative results that intuitively support this conjecture:
Vassilevska W. and Williams proved that a truly subcubic combinatorial algorithm for $3$-Clique implies such algorithm for BMM as well \cite{VW10}.
Unconditional lower bounds for $k$-Clique are known for various computational models, such as $\Omega(n^k)$ for monotone circuits \cite{AB87}.
The planted Clique problem has also proven to be very challenging (e.g. \cite{AlonAKMRX07,AlonKS98,HK11,Jerrum92}).
Max-Clique is also known to be hard to efficiently approximate within nontrivial factors \cite{hastad}.

Formally, our reduction from $k$-Clique to CFG recognition proves the following theorem.

\begin{theorem}
\label{thm:cfg}
There is context-free grammar $\mathcal{G}_C$ of constant size such that if we can determine if a string of length $n$ can be obtained from $\mathcal{G}_C$ in $T(n)$ time, then $k$-Clique on $n$ node graphs can be solved in $O\left(T \left(n^{k/3+1} \right) \right)$ time, for any $k\geq 3$. 
Moreover, the reduction is combinatorial.
\end{theorem}

To see the tightness of our reduction, let $1 \leq f \leq \omega$ and $1 \leq c \leq 3$ denote the smallest numbers such that CFG recognition can be solved in $O(n^f)$ time and combinatorially in $O(n^c)$ time.
An immediate corollary of Theorem~\ref{thm:cfg} is that $f\geq F$ and $c\geq C$.
Under the plausible assumption that current $k$-Clique algorithms are optimal, up to $n^{o(1)}$ improvements,  our theorem implies that $f \geq \omega$ and $c \geq 3$.
Combined with Valiant's algorithm we get that $f = \omega$ and with standard CFG parsers we get that $c=3$.
Because our grammar size $g$ is fixed, we also rule out $O(h(g) \cdot n^{3-\eps})$ time combinatorial CFG parsers for any computable function $h(g)$.

In other words, we construct a single fixed context-free grammar $\mathcal{G}_C$ for which the recognition problem (and therefore any parsing problem) cannot be solved any faster than by Valiant's algorithm and any combinatorial recognizer will not run in truly subcubic time, without implying a breakthrough algorithm for the Clique problem.
This (conditionally) proves that these algorithms are optimal general purpose CFG parsers, and more efficient parsers will only work for CFL with special restricting properties. 
On the positive side, our reduction might hint at what a CFG should look like to allow for efficient parsing.

The \emph{online} version of CFG recognition is as follows: preprocess a CFG such that given a string $w$ that is revealed one letter at a time, so that at stage $i$ we get $w[1\cdots i]$, we can say as quickly as possible whether $w[1\cdots i]$ can be derived from the grammar (before seeing the next letters).
One usually tries to minimize the total time it takes to provide all the $|w|=n$ answers.
This problem has a long history of algorithms \cite{Wei76,Graham80,Rytter85} and lower bounds \cite{HK65,Gal69,Sei86}.
The current best upper bound is $O(n^3/\log^2 n)$ total running time, and the best lower bound is $\Omega(n^2/\log{n})$.
Since this is a harder problem, our lower bound for CFG recognition also holds for it.


\subsection{More Results}

The main ingredient in the proof of Theorem~\ref{thm:cfg} is a lossless encoding of a graph into a string that belongs to a simple CFL iff the graph contains a $k$-Clique.
Besides classifying the complexity of a fundamental problem, this construction has led us to new lower bounds for two more modern and well-studied cubic time problems for which faster algorithms are highly desirable in practice.

\paragraph{RNA Folding}
The RNA folding problem is a version of maximum matching and is one of the most relevant problems in computational biology.
Its most basic version can be neatly defined as follows.
Let $\Sigma$ be a set of letters and let $\Sigma' = \{ \sigma' \mid \sigma \in \Sigma \}$ be the set of ``matching" letters, such that for every letter $\sigma \in \Sigma$ the pair $\sigma,\sigma'$ match.
Given a sequence of $n$ letters over  $\Sigma \cup \Sigma'$ the RNA folding problem asks for the maximum number of \emph{non-crossing} pairs $\{i,j\}$ such that the $i^{th}$ and $j^{th}$ letter in the sequence match.

The problem can be viewed as a \emph{Stochastic CFG} parsing problem in which the CFG is very restricted.
This intuition has led to many mildly subcubic algorithms for the problem \cite{RNAsub1,akutsu,RNAsub2,RNAsub3,RNAsub4,RNAsub5,ZTZ10}.
The main idea is to adapt Valiant's algorithm by replacing his BMM with a $(min,+)$-matrix product computation (i.e. distance or tropical product).
Using the fastest known algorithm for $(min,+)$-matrix multiplication it is possible to solve RNA folding in $O(n^3/2^{\sqrt{\log{n}}})$ time \cite{ryan-apsp}.
The fastest \emph{combinatorial} algorithms, however, run in roughly $O(n^3/\log^2{n})$ time \cite{RNAsub1,RNAsub2}, and we show that a truly subcubic such algorithm would imply a breakthrough Clique algorithm.
Our result gives a negative partial answer to an open question raised by Andoni \cite{bertinoro}.

\begin{theorem}
\label{thm:RNA}
If RNA Folding on a sequence of length $n$ can be solved in $T(n)$ time, then $k$-Clique on $n$ node graphs can be solved in $O\left(T \left(n^{k/3+O(1)} \right) \right)$ time, for any $k\geq 3$. 
Moreover, the reduction is combinatorial.
\end{theorem}

Besides a tight lower bound for combinatorial algorithms, our result also shows that a faster than $O(n^{\omega})$ algorithm for RNA Folding is unlikely.
Such an upper bound is not known for the problem, leaving a small gap in the complexity of the problem after this work.
However, we observe that the known reductions from RNA Folding to $(min,+)$-matrix multiplication produce matrices with very special ``bounded monotone" structure: the entries are bounded by $n$ and every row and column are monotonically increasing.
This exact structure has allowed Chan and Lewenstein to solve the ``bounded monotone" $(min,+)$-convolution problem in truly subquadratic time \cite{add_comb}.
Their algorithm uses interesting tools from additive combinatorics and it seems very plausible that the approach will lead to a truly subcubic (non-combinatorial) algorithm for ``bounded monotone" $(min,+)$ product and therefore for RNA Folding.

\paragraph{Dyck Edit Distance}
The lower bound in Theorem~\ref{thm:cfg} immediately implies a similar lower bound for the \emph{Language Edit Distance} problem on CFLs, in which we want to be able to return the minimal edit distance between a given string $w$ and a string in the language.
That is, zero if $w$ is in the language and the length of the shortest sequence of insertions, deletions, substitutions that is needed to convert $w$ to a string that is in the language otherwise.
This is a classical problem introduced by Aho and Peterson in the early 70's with cubic time algorithms \cite{AP,Myers} and many diverse applications \cite{Korn13,G+00,Fis80}.
Very recently, Rajasekaran and Nicolae \cite{RN14} and Saha~\cite{SahaArxiv} obtained truly subcubic time \emph{approximation} algorithms for the problem for arbitrary CFGs.

However, in many applications the CFG we are working with is very restricted and therefore easy to parse in linear time.
One of the simplest CFGs with big practical importance is the Dyck grammar which produces all strings of well-balanced parenthesis.
The Dyck recognition problem can be easily solved in linear time with a single pass on the input.
Despite the grammar's very special structure, the Dyck Edit Distance problem is not known to have a subcubic algorithm.
In a recent breakthrough, Saha \cite{SahaDyckApprox} presented a near-linear time algorithm that achieves a logarithmic approximation for the problem.
Dyck edit distance can be viewed as a generalization of the classical string edit distance problem whose complexity is essentially quadratic \cite{clrs,edit_hardness}, and Saha's approximation algorithm nearly matches the best known approximation algorithms for string edit distance of Andoni, Krauthgamer, and Onak \cite{AOK10}, both in terms of running time and approximation factor.
This naturally leads one to wonder whether the complexity of (exact) Dyck edit distance might be also quadratic.
We prove that this is unlikely unless $\omega = 2$ or there are faster clique finding algorithms.

\begin{theorem}
\label{thm:dyck}
If Dyck edit distance on a sequence of length $n$ can be solved in $T(n)$ time, then $3k$-Clique on $n$ node graphs can be solved in $O\left(T \left(n^{k+O(1)} \right) \right)$ time, for any $k\geq 1$. 
Moreover, the reduction is combinatorial.
\end{theorem}

Our result gives an answer to an open question of Saha \cite{SahaDyckApprox}, who asked if Lee's lower bound holds for the Dyck Edit Distance problem, and shows that the search for good approximation algorithms for the problem is justified since efficient exact algorithms are unlikely.


\paragraph{Remark}
A simple observation shows that the \emph{longest common subsequence} (LCS) problem on two sequences $x,y$ of length $n$ over an alphabet $\Sigma$ can be reduced to RNA folding on a sequence of length $2n$:
 if $y=y_1 \ldots y_n$ then let $\hat{y}:=y_n' \ldots y_1'$ and then $RNA(x \circ \hat{y})=LCS(x,y)$.
A quadratic lower bound for LCS was recently shown under the Strong Exponential Time Hypothesis (SETH) \cite{ABV15,BK15}, which implies such a lower bound for RNA folding as well (and, with other ideas from this work, for CFG recognition and Dyck Edit Distance).
However, we are interested in higher lower bounds, ones that match Valiant's algorithm and basing such lower bounds on SETH would imply that faster matrix multiplication algorithms refute SETH - a highly unexpected breakthrough.
Instead, we base our hardness on $k$-Clique and devise more delicate constructions that use the cubic-time nature of our problems.

\paragraph{Proof Outlines}
In our three proofs, the main approach is the following.
We will first preprocess a graph $G$ in $O(n^{k+O(1)})$ time in order to construct an encoding of it into a string of length $O(n^{k+O(1)})$.
This will be done by enumerating all $k$-cliques and representing them with carefully designed gadgets such that a triple of clique gadgets will ``match well" if and only if the triple make a $3k$-clique together, that is, if all the edges between them exists.
We will use a fast (subcubic) CFG recognizer, RNA folder, or an Dyck Edit Distance algorithm to speed up the search for such a ``good" triple and solve $3k$-Clique faster than $O(n^{3k})$.
These clique gadgets will be constructed in similar ways in all of our proofs.
The main differences will be in the combination of these cliques into one sequence.
The challenge will be to find a way to combine $O(n^k)$ gadgets into a string in a way that a ``good" triple will affect the overall score or parse-ability  of the string.


\paragraph{Notation and Preliminaries}
All graphs in this paper will be on $n$ nodes and $O(n^2)$ undirected and unweighted edges.
We associate each node with an integer in $[n]$ and let $\bar{v}$ denote the encoding of $v$ in binary and we will assume that it has length exactly $2\log{n}$ for all nodes in $V(G)$.
When a graph $G$ is clear from context, we will denote the set of all $k$-cliques of $G$ by $\mathcal{C}_k$.
We will denote concatenation of sequences with $x \circ y$, and the reverse of a sequence $x$ by $x^R$.
Problem definitions and additional problem specific preliminaries will be given in the corresponding section.

%% file: proof.tex
\def \start {\text{start}}
\def \endd {\text{end}}
\def \midd {\text{mid}}

This section we show our reduction from Clique to CFG recognition and prove Theorem~\ref{thm:cfg}.


Given a graph $G=(V,E)$, we will construct a string $w$ of length $O(k^2 \cdot n^{k+1})$ that encodes $G$.
The string will be constructed in $O(k^2 \cdot n^{k+1})$ time which is linear in its length.
Then, we will define our context free grammar $\mathcal{G}_C$ which will be independent of $G$ or $k$ and it will be of constant size, such that our string $w$ will be in the language defined by $\mathcal{G}_C$ if and only if $G$ contains a $3k$ clique.
This will prove Theorem~\ref{thm:cfg}.

Let $\Sigma = \{ \mathtt{0,1, \$, \# , a_{\start}, a_{\midd}, a_{\endd}, b_{\start}, b_{\midd}, b_{\endd}, c_{\start}, c_{\midd}, c_{\endd}} \}$  be our set of $13$ terminals (alphabet).
As usual, $\eps$ will denote the empty string.
We will denote the derivation rules of a context free grammar with $ \to $  and the derivation (by applying one or multiple rules) with $\implies$.


\paragraph{The string}
First, we will define \emph{node} and \emph{list} gadgets:
 \[ 
 NG(v) = \# \  \bar{v} \  \#  \ \ \ \ \text{  and   }  \ \ \ \ 
 LG(v) = \# \ \bigcirc_{u \in N(v)} (\$ \ \bar{u}^R \ \$) \ \#
 \]
 
Consider some $t = \{ v_1,\ldots,v_k\} \in \mathcal{C}_k$.
We now define ``clique node" and ``clique list" gadgets.
\[
CNG(t) =  \bigcirc_{v \in t}  (NG(v))^k \ \ \text {and } \ \ CLG(t) =  (\bigcirc_{v \in t}  LG(v))^k
\]
and our main clique gadgets will be:
\[
CG_\alpha(t) = \mathtt{a_{\start}} \ CNG(t) \  \mathtt{a_{\midd}} \ CNG(t) \ \mathtt{a_{\endd}}
\]
\[
CG_\beta(t) = \mathtt{b_{\start}} \ CLG(t) \  \mathtt{b_{\midd}} \ CNG(t) \ \mathtt{b_{\endd}}
\]
\[
CG_\gamma(t) = \mathtt{c_{\start}} \ CLG(t) \  \mathtt{c_{\midd}} \ CLG(t) \ \mathtt{c_{\endd}}
\]

Finally, our encoding of a graph into a sequence is the following:  
\[
w =  \left( \bigcirc_{t \in \mathcal{C}_k }  CG_\alpha(t) \right)
\left( \bigcirc_{t \in \mathcal{C}_k}  CG_\beta(t) \right)
\left( \bigcirc_{t \in \mathcal{C}_k}  CG_\gamma(t) \right)
\]


\paragraph{The Clique Detecting Context Free Grammar.}
The set of non-terminals in our grammar $\mathcal{G}_C$ is: 

$$\mathcal{T} = \{ \mathbf{
S,W,W',V,
S_{\alpha\gamma}, S_{\alpha\beta}, S_{\beta\gamma}, 
S^{\star}_{\alpha\gamma}, S^{\star}_{\alpha\beta}, S^{\star}_{\beta\gamma}, 
N_{\alpha\gamma}, N_{\alpha\beta}, N_{\beta\gamma}
}
\}.
$$
 The ``main" rules are:
\begin{align*}
&\mathbf{S}  \to \  \mathbf{W} \ \mathtt{a_{\start}}  \  \mathbf{S_{\alpha\gamma}} \  \mathtt{c_{\endd}}  \  \mathbf{W} \\
&\mathbf{S^{\star}_{\alpha\gamma}}  \to \  \mathtt{a_{\midd}} \ \mathbf{S_{\alpha\beta}} \ \mathtt{b_{\midd}} \mathbf{S_{\beta\gamma}} \ \mathtt{c_{\midd}}   \\
&\mathbf{S^{\star}_{\alpha\beta}} \to \  \mathtt{a_{\endd}} \ \mathbf{W} \ \mathtt{b_{\start}}   \\
&\mathbf{S^{\star}_{\beta\gamma}}  \to \  \mathtt{b_{\endd}} \ \mathbf{W} \ \mathtt{c_{\start}}  
\end{align*}

And for every $xy \in \{ \alpha\beta, \alpha\gamma,\beta\gamma\}$ we will have the following rules in our grammar.
These rules will be referred to as ``listing" rules.
\begin{align*}
& \mathbf{S}_{xy} \to \  \mathbf{S^{\star}}_{xy} \\
& \mathbf{S}_{xy} \to \  \# \ \mathbf{N}_{xy} \ \$ \ \mathbf{V} \ \# \\
& \mathbf{N}_{xy} \to \   \#  \  \mathbf{S}_{xy} \ \# \ \mathbf{V} \ \$  \\
& \mathbf{N}_{xy} \to \ \mathtt{\sigma} \  \mathbf{N}_{xy} \ \mathtt{\sigma} \ & \forall \mathtt{\sigma} \in \{\mathtt{0,1}\} 
\end{align*}

Then we also add ``assisting" rules:
\begin{align*}
& \mathbf{W} \to \ \mathtt{\eps} \ \ | \  \sigma \  \mathbf{W} \ & \forall \sigma \in \Sigma \\
& \mathbf{W'} \to \ \eps \ | \  \sigma \  \mathbf{W'} \ & \forall \sigma \in \{\mathtt{0,1}\} \\
& \mathbf{V} \to  \ \ \eps  \ \ |  \ \$ \ \mathbf{W'}  \ \$ \ \mathbf{V} &
\end{align*}

Our Clique Detecting grammar $\mathcal{G}_C$ has $13$ non-terminals $\mathcal{T}$, $13$ terminals $\Sigma$, and $38$ derivation rules.
The \emph{size} of $\mathcal{G}_C$, i.e. the sum of the lengths of the derivation rules, is $132$.

\paragraph{The proof.}
This proof is essentially by following the derivations of the CFG, starting from the starting symbol $\mathbf{S}$ and ending at some string of terminals, and showing that the resulting string must have certain properties.
Any encoding of a graph into a string as we describe will have these properties iff the graph contains a $3k$-clique. 
The correctness of the reduction will follow from the following two claims.

\begin{claim}
\label{cfg:hard}
If $\mathcal{G}_C\implies w$ then $G$ contains a $3k$-clique.
\end{claim}

\begin{proof}
The derivation of $w$ must look as follows.
First we must apply the only starting rule,
$$\mathbf{S}\ \implies \ w_1 \mathtt{\ a_{\start}}  \ \mathbf{S_{\alpha\gamma}} \mathtt{\ c_{\endd}} \ w_2$$
where $\mathtt{a_{\start}}$ appears in $CG_\alpha(t_\alpha)$ for some $t_\alpha \in \mathcal{C}_k$
and $\mathtt{c_{\start}}$ appears in $CG_\gamma(t_\gamma)$ for some $t_\gamma \in \mathcal{C}_k$,
and $w_1$ is the prefix of $w$ before $CG(t_\alpha)$ and $w_2$ is the suffix of $w$ after $CG(t_\gamma)$.
Then we can get,
$$\mathbf{S_{\alpha\gamma}} \implies CNG(t_\alpha)  \ \mathbf{S^{\star}_{\alpha\gamma}} \ CLG(t_\gamma)$$
by repeatedly applying the $xy$-``listing" rules  where $xy=\alpha\gamma$ and finally terminating with the rule $\mathbf{S_{\alpha\gamma} \to S^{\star}_{\alpha\gamma}}$.
By Lemma~\ref{biclique} below, this derivation is only possible if the nodes of $t_\alpha \cup t_\gamma$ make a $2k$-clique (call this observation (*)).
Then we have to apply the derivation: 
$$\mathbf{S^{\star}_{\alpha\gamma}} \ \implies \mathtt{\ a_{\midd}} \ \mathbf{S_{\alpha\beta}} \ \mathtt{b_{\midd}} \mathbf{S_{\beta\gamma}} \ \mathtt{c_{\midd}} $$
and for some $t_\beta \in \mathcal{C}_k$ we will have,
\begin{center}
$\mathbf{S_{\alpha\beta}} \implies \ CNG(t_\alpha) \ \mathbf{S^{\star}_{\alpha\beta}} \ CLG(t_\beta), \ $ and $ \ \mathbf{S_{\beta\gamma}} \implies \ CNG(t_\beta) \ \mathbf{S^{\star}_{\beta\gamma}} \ CLG(t_\gamma)$,
\end{center}
where in both derivations we repeatedly use ``listing" rules before exiting with the $\mathbf{S}_{xy} \to \mathbf{S^{\star}}_{xy}$ rule.
Again, by Lemma~\ref{biclique} we get that the nodes of $t_\alpha \cup t_\beta$ are a $2k$ clique in $G$, and that the nodes of $t_\beta \cup t_\gamma$ are a $2k$ clique as well (call this observation (**)).
Finally, we will get the rest of $w$ using the derivations:
\begin{center}
$\mathbf{S^{\star}}_{\alpha\beta} \implies \ \mathtt{a_{\endd}} \ w_3 \  \mathtt{b_{\start}}, \ $ and $ \ \mathbf{S^{\star}}_{\beta\gamma} \implies \mathtt{\ b_{\endd}} \ w_4 \ \mathtt{c_{\start}}$,
\end{center}
where $w_3$ is the substring of $w$ between $CG(t_\alpha)$ and $CG(t_\beta)$, and similarly $w_4$ is the substring of $w$ between $CG(t_\beta)$ and $CG(t_\gamma)$.

Combining observations (*) and (**), that we got from the above derivation scheme and Lemma~\ref{biclique}, we conclude that the nodes of $t_\alpha \cup t_\beta \cup t_\gamma$ form a $3k$-clique in $G$, and we are done.

To complete the proof we will now prove Lemma~\ref{biclique}.

\begin{lemma}
\label{biclique}
If for some $t,t' \in \mathcal{C}_k$ and $xy \in \{\alpha\beta,\alpha\gamma,\beta\gamma\}$ we can get the derivation $\mathbf{S}_{xy} \implies CNG(t) \ \mathbf{S^{\star}}_{xy} \ CLG(t'), \ $ only using the ``listing" rules, then $t \cup t'$ forms a $2k$ clique in $G$.
\end{lemma}

\begin{proof}
Any sequence of derivations starting at $\mathbf{S}_{xy}$ and ending at $\mathbf{S^{\star}}_{xy}$ will have the following form.
From $\mathbf{S}_{xy}$ we can only proceed to non-terminals $\mathbf{V}$ and $\mathbf{N}_{xy}$.
The non-terminal $\mathbf{V}$ does not produce any other non-terminals.
A single instantiation of $\mathbf{S}_{xy}$ can only produce (via the second ``listing" rule) a single instantiation of $\mathbf{N}_{xy}$.
In turn, from $\mathbf{N}_{xy}$ we can only proceed to non-terminals $\mathbf{V}$ and $\mathbf{S}_{xy}$, and again, a single instantiation of $\mathbf{N}_{xy}$ can produce a single instantiation of $\mathbf{S}_{xy}$.
Thus, we produce some terminals (on the left and on the right) from the set $\mathtt{\{\#, \$, 0,1\}}$ and then we arrive to $\mathbf{S}_{xy}$ again.
This can repeat an arbitrary number of times, until we apply the rule $\mathbf{S}_{xy} \to \mathbf{S^{\star}}_{xy}$.

Thus, the derivations must look like this:
	$$
		\mathbf{S}_{xy} \implies \ell \ \mathbf{S}^{\star}_{xy} \  r
	$$
	for some strings $\ell,r$ of the form $\mathtt{\{\#,\$,0,1\}^*}$,
 and our goal is to prove that $\ell,r$ satisfy certain properties.

	It is easy to check that $\mathbf{V}$ can derive strings of the following form $p= \mathtt{(\$ \ \{0,1\}^* \ \$)^*}$, that is, it produces a list (possibly of length $0$) of binary sequences (possibly of length $0$) surrounded by $\$$ symbols (between every two neighboring binary sequences there are two $\$$). 
	A key observation is that repeated application of the fourth ``listing'' rule gives derivations $\mathbf{N}_{xy}\implies s \ \mathbf{N}_{xy}  \ s^R$, for any $s \in \mathtt{\{0,1\}^*}$. 
	Combining these last two observations, we see that when starting with $\mathbf{S}_{xy}$ we can only derive strings of the following form, or terminate via the rule $\mathbf{S}_{xy} \to \mathbf{S^{\star}}_{xy}$.
	\begin{equation} \label{eq:inlists}
		\mathbf{S}_{xy}   \implies \  \#  \ \mathbf{N}_{xy} \ \$ \ p_1 \ \#  \implies \  \#  \ s  \ \mathbf{N}_{xy}  \ s^R \ \$ \ p_1 \ \#  \implies \  \#  \ s \ \# \ \mathbf{S}_{xy} \ \# \ p_2 \ \$ \ s^R \ \$ \ p_1 \ \#
	\end{equation}
	for some $p_1,p_2$ of the form $ \mathtt{(\$ \ \{0,1\}^* \ \$)^*}$.
	
	Now consider the assumption in the statement of the lemma, and recall our constructions of ``clique node gadget" and ``clique list gadget" .
	By construction, $CNG(t)$ is composed of $k^2$ node gadgets ($NG$) separated by \# symbols, and CLG is composed of $k^2$ list gadgets ($LG$) separated by \# symbols.
	Note also that the list gadgets contain $O(n)$ node gadgets within them and those are separated by \$ symbols, and there are no \# symbols within the list gadgets.

	For every $i \in [k^2]$, let $\ell_i$ be the $i^{th}$ $NG$ in $CNG(t)$ and let $r_i$ be the $i^{th}$ $LG$ in $CLG(t')$.
	Then, for every $2\leq i \leq k$, in the derivation $\mathbf{S}_{xy} \implies CNG(t) \ \mathbf{S}^{\star}_{xy} \ CLG(t')$, we must have had the derivation 
	\[
	 \mathbf{S}_{xy} \implies \ell_1 \ \cdots \ \ell_{i-1} \ ( \ \mathbf{S}_{xy} \ ) \ r_{k^2-i+2} \ \cdots \ r_{k^2} \ \implies \ \ell_1 \ \cdots \ \ell_{i-1} \ ( \ \ell_i \ \mathbf{S}_{xy} \ r_{k^2-i+1} \ ) \ r_{k^2-i+2} \ \cdots \ r_{k^2}
	 \]
	 and by (1) this implies that the binary encoding of the node $v\in t$ that appears in the $i^{th}$ $NG$ in $CNG(t)$ must appear in one of the $NG$s that appear in the $(k^2-i+1)$ $LG$ in $CLG(t')$ which corresponds to a node $u \in t'$.
	Since $LG(u)$ contains a list of neighbors of the node $u$, this implies that $v \in N(u)$ and $\{u,v\} \in E$.
	Also note that $u \notin N(u)$ and therefore $u$ does not appear in $LG(u)$ and therefore $v$ cannot be equal to $u$ if this derivation occurs.
	
	Now, consider any pair of nodes $v \in t, u \in t'$.
	By the construction of CNG and CLG, we must have an index $i \in [k^2]$ such that the $i^{th}$ NG in $CNG(t)$ is $NG(v)$ and the $(k^2-i+1)$ LG in $CLG(t')$ is $CLG(u)$.
	By the previous argument, we must have that $u \neq v$ and $\{u,v\} \in E$ is an edge.	
	Given that $t,t'$ are $k$-cliques themselves, and any pair of nodes $v \in t, u \in t'$ must be neighbors (and therefore different), we conclude that $t \cup t'$ is a $2k$-clique.
\end{proof}
\end{proof}

\begin{claim}
\label{cfg:easy}
If $G$ contains a $3k$-clique, then $\mathcal{G}_C\implies w$.
\end{claim} 

\begin{proof}
This claim follows by following the derivations in the proof of Claim~\ref{cfg:hard} with any triple $t_\alpha,t_\beta,t_\gamma \in \mathcal{C}_k$ of $k$-cliques that together form a $3k$-clique.
\end{proof}

We are now ready to prove Theorem~\ref{thm:cfg}.

\begin{reminder}{Theorem~\ref{thm:cfg}}
There is context-free grammar $\mathcal{G}_C$ of constant size such that if we can determine if a string of length $n$ can be obtained from $\mathcal{G}_C$ in $T(n)$ time, then $k$-Clique on $n$ node graphs can be solved in $O\left(T \left(n^{k/3+1} \right) \right)$ time, for any $k\geq 3$. 
Moreover, the reduction is combinatorial.
\end{reminder}

\begin{proof}
Given an instance of $3k$-Clique, a graph $G=(V,E)$ we construct the string $w$ as described above, which will have length $O(k^2 \cdot n^{k+1})$, in $O(k^2 \cdot n^{k+1})$ time.
Given a recognizer for $\mathcal{G}_C$ as in the statement, we can check whether $\mathcal{G}_C \implies w$ in $O\left(T \left(n^{k/3+1} \right) \right)$ time (treating $k$ as a constant).
By Claims~\ref{cfg:hard} and~\ref{cfg:easy}, $\mathcal{G}_C \implies w$ iff the graph $G$ contains a $3k$-clique.
\end{proof}

%% file: rna.tex
In this section we prove Theorem~\ref{thm:RNA} by reducing $k$-Clique to RNA folding, defined below.

Let $\Sigma$ be an alphabet of letters of constant size.
For any letter $\sigma \in \Sigma$ there will be exactly one ``matching" letter which will be denoted by $\sigma'$.
Let $\Sigma' = \{ \sigma' \mid \sigma \in \Sigma\}$ be the set of matching letters to the letters in $\Sigma$.
Throughout this section we will say that a pair of letters $\{x,y\}$ match iff $y=x'$ or $x=y'$.

Two pairs of indices $(i_1,j_1),(i_2,j_2)$ such that $i_1<j_1$ and $i_2<j_2$ are said to ``cross" iff at least one of the following three conditions hold:
(i) $i_1=i_2$ or $i_1=j_2$, or $j_1=i_2$, or $j_1=j_2$; (ii) $i_1<i_2<j_1<j_2$; (iii) $i_2<i_1<j_2<j_1$. Note that by our definition, non-crossing pairs cannot share any indices. 

\begin{definition}[RNA Folding]
Given a sequence $S$ of $n$ letters from $\Sigma \cup \Sigma'$, what is the maximum number of pairs $A = \{(i,j) \mid i < j\text{ and }i,j \in [n]\}$ such that for every pair $(i,j) \in A$ the letters $S[i]$ and $S[j]$ match and there are no crossing pairs in $A$. 
We will denote this maximum value by $RNA(S)$.
\end{definition}

It is interesting to note that RNA can be seen as a language distance problem with respect to some easy to parse grammar.
Because of the specific structure of this grammar, our reduction from Section~\ref{sec:cfg} does not apply.
However, the ideas we introduced allow us to replace our clique detecting grammar with an easier grammar if we ask the parser to return more information, like the distance to a string in the grammar.
At a high level, this is how we get the reduction to RNA folding presented in this section.

To significantly simplify our proofs, we will reduce $k$-Clique to a more general \emph{weighted} version of RNA folding.
Below we show that this version can be reduced to the standard RNA folding problem with a certain overhead.

\begin{definition}[Weighted RNA Folding]
Given a sequence $S$ of $n$ letters from $\Sigma \cup \Sigma'$ and a weight function $w:\Sigma \to [M]$, what is the \emph{maximum weight} of a set of pairs $A = \{(i,j) \mid i < j\text{ and }i,j \in [n]\}$ such that for every pair $(i,j) \in A$ the letters $S[i]$ and $S[j]$ match and there are no crossing pairs in $A$. 
The weight of $A$ is defined as $\sum_{(i,j) \in A} w(S[i])$.
We will denote this maximum value by $WRNA(S)$.
\end{definition}

\begin{lemma}
\label{WRNA}
An instance $S$ of Weighted RNA Folding on a sequence of length $n$, alphabet $\Sigma \cup \Sigma'$, and weight function $w:\Sigma \to [M]$ can be reduced to an instance $\hat S$ of RNA Folding on a sequence of length $O(Mn)$ over the same alphabet.
\end{lemma}

\begin{proof}
	Let $S=S_1 \cdots S_n$, we set $\hat S:=S_1^{w(S_1)} \cdots S_n^{w(S_n)}$, that is, each symbol $S_i$ is repeated $w(S_i)$ times.
	First, we can check that $WRNA(S)\leq RNA(\hat{S})$. This holds because we can replace each matching pair $\{a,a'\}$ in the folding achieving weighted RNA score of $WRNA(S)$ with $w(a)$ such pairs in the (unweighted) RNA folding instance $\hat{S}$ giving the same contribution to $RNA(\hat{S})$.
	

	Now we will show that $RNA(\hat{S})\leq WRNA(S)$. 
	Suppose thet there are symbols $a$ and $a'$ in $\hat S$ that are paired.
	The symbol $a$ comes from a sequence $s_1$ of $a$ symbols of length
	$w(a)$. The sequence $s_1$ was produced from a single symbol $a$ when transforming $S$ into $\hat{S}$.
	Similarly, the symbol $a'$ comes from a sequence $s_2$ of $a'$ symbols of length $w(a)$.
	Also, assume that there exists a symbol in $s_1$ that is paired to a symbol that is outside of $s_2$ or there exists a symbol in $s_2$ that is paired to a symbol outside of $s_1$.
	While we can find such symbols $a$ and $a'$, we repeat the following procedure.
	Choose $a$ and $a'$ that satisfy the above properties. And choose them so that the number of other symbols between $a$ and $a'$ is as small as possibly. Break ties arbitrarily.
	We match all symbols in $s_1$ to their counterparts in $s_2$. Also, we rematch all symbols that were previously matched to $s_1$ or $s_2$ among themselves.
	We can check that we can rematch these symbols so that the number of matched pairs do not decrease.

	Therefore, we can assume that in some optimal folding of $\hat{S}$, for any pair $a\in \Sigma,a' \in \Sigma'$ that is matched the corresponding substrings $s_1$ and $s_2$ are completely paired up. 
	Thus, to get a folding of $S$ that achieves $WRNA(S)$ at least $RNA(\hat{S})$ we can now simply fold the corresponding symbols to $s_1$ and $s_2$, for any such pair $a,a'$.
	\end{proof}

\subsection{The Reduction}

Given a graph $G=(V,E)$ on $n$ nodes and $O(n^2)$ unweighted undirected edges, we will describe how to efficiently construct a sequence $S_G$ over an alphabet $\Sigma$ of constant size, such that the RNA score of $S_G$ will depend on whether $G$ contains a $3k$-clique.
The length of $S_G$ will be $O(k^d n^{k+c})$ for some small fixed constants $c,d>0$ independent of $n$ and $k$, and the time to construct it from $G$ will be linear in its length.
This will prove that a fast (e.g. subcubic) RNA folder can be used as a fast $3k$-clique detector (one that runs much faster than in $O(n^{3k})$ time).

Our main strategy will be to enumerate all $k$-cliques in the graph and then search for a triple of $k$-cliques that have all the edges between them.
We will be able to find such a triple iff the graph contains a $3k$-clique.
An RNA folder will be utilized to speed up the search for such a ``good" triple.
Our reduction will encode every $k$-clique of $G$ using a ``short" sequence of length $O(n^c)$ such that the RNA folding score of a sequence composed of the encodings of a triple of sequences will be large iff the triple is ``good".
Then, we will show how to combine the short encodings into our long sequence $S_G$ such that the existence of a ``good" triple affects the overall score of an optimal folding.

\paragraph{The RNA Sequence}
Our sequence $S_G$ will be composed of many smaller gadgets which will be combined in certain ways by other padding gadgets.
We construct these gadgets now and explain their useful properties.
The proofs of these properties are postponed until after we present the whole construction of $S_G$.

For a sequence $s \in \Sigma^*$ let $p(s) \in (\Sigma')^*$ be the sequence obtained from $s$ by replacing every letter $\sigma \in \Sigma$ with the matching letter $\sigma' \in \Sigma'$.
That is, if $s=s_1\cdots s_n$ then $p(s)=s_1'\cdots s_n'$.

 Our alphabet $\Sigma$ will contain the letters $\mathtt{0,1}$ and some additional symbols which we will add as needed in our gadgets.
 We will set the weights so that $w(\mathtt{0})=w(\mathtt{1})=1$, and the extra symbols we add will be more ``expensive".
 For example, we will add the $\$$ symbol to the alphabet and set $w(\$)=10\cdot \log{n}$.
 We define \emph{node} gadgets as, 
 \[ 
 NG(v) = \$^{2n} \  \bar{v} \  \$^{2n}
 \]
 and \emph{list} or \emph{neighborhood} gadgets as,
 \[
 LG(v) =   \bigcirc_{u \in N(v)} (\$ \ \bar{u} \ \$) \ \circ \ \bigcirc_{u \notin N(v)} (\$ \ \$).
 \]
 
These gadget are constructed so that for any two nodes $u,v \in V(G)$, the RNA folding score of the sequence $NG(v) \circ p\left(LG(u) \right)^R$ is large (equal to some fixed value $E_1$) if $v$ is in the neighborhood of $u$, that is $(u,v) \in E(G)$, and smaller otherwise (at most $E_1-1$).
This is because the more expensive $\$$ symbols force an optimal folding to match $\bar{v}$ with exactly one $p(\bar{w})^R$, since otherwise a $\$'$ will remain unmatched while many $\$$ are free.
The construction also allows the folding to pick any $w \in N(u)$ to fold together with $\bar{v}$, without affecting the score from the $\$$ symbols.
Then, we use the fact that $\bar{v} \circ p(\bar{w})^R$ achieves maximal score iff $v=w$.
 This is formally proved in Claim~\ref{nbrhood}.
 
Let $\ell_1 = 10 k^2 \cdot n\log{n}$ and note that $\ell_1$ is an upper bound on the total weight of all the symbols in the gadgets $NG(v)$ and $LG(v)$, for any node $v \in V(G)$.
Let $\mathcal{C}_k$ be the set of $k$-cliques in $G$ and consider some $t = \{ v_1,\ldots,v_k\} \in \mathcal{C}_k$.
We will now combine the node and list gadgets into larger gadgets that will be encoding $k$-cliques.

We will add the $\#$ symbol to the alphabet and set $w(\#)= \ell_1$, i.e. a single $\#$ letter is more expensive than entire $k^2$ node or list gadget.
We will encode a clique in two ways. The first one is,
\[
CNG(t) =  \bigcirc_{v \in t}  (\# NG(v) \#)^k 
\]
and the second one is,
\[
CLG(t) =  (\bigcirc_{v \in t}  (\# LG(v) \# ))^k .
\]
These clique gadgets are very useful because of the following property.
For any two $k$-cliques $t_1,t_2 \in \mathcal{C}_k$, the RNA folding score of the sequence $CNG(t_1) \circ p \left( CLG(t_2) \right)^R$ is large (equal to some fixed value $E_2$) if $t_1$ and $t_2$, together,  form a $2k$-clique, and is smaller otherwise (at most $E_2-1$).
That is, the RNA folding score of the sequence tells us whether any pair of nodes $u \in t_1, v \in t_2$ are connected $(u,v)\in E(G)$.
There are two ideas in the construction of these gadgets.
First, we copy the gadgets corresponding to the $k$ nodes of the cliques $k$ times, resulting in $k^2$ gadgets, and we order them in a way so that for any pair of nodes $u\in t_1,v\in t_2$ there will be a position $i$ such that the gadget of $u$ in $CNG(t_1)$ and the gadget of $v$ in $p(CLG(t_2))^R$ are both at position $i$.
Then, we use the expensive $\#$ separators to make sure that in an optimal RNA folding of $CNG(t_1)$ and $p(CLG(t_2))^R$, the gadgets at positions $i$ are folded together, and not to other gadgets - otherwise some $\#$ symbol will not be paired.
 This is formally proved in Claim~\ref{twoCG}.

Let $\ell_2 = 10 \cdot k^2 \cdot \ell_1= O(n\log{n})$ and note that it is an upper bound on the total weight of all the symbols in the $CNG(t)$ and $CLG(t)$ gadgets.
Finally, we introduce a new letter to the alphabet $\mathtt{g}$ and set its weight to $w(\mathtt{g}) = \ell_2$, which is much more expensive than the entire gadgets we constructed before, and then define our final clique gadgets.
Moreover, we will now duplicate our alphabet three times  to force only ``meaningful" foldings between our gadgets.
It will be convenient to think of $\alpha,\beta,\gamma$ as three \emph{types} such that we will be looking for three $k$-clique, one from type $\alpha$ one from $\beta$ and one from $\gamma$.
For any  pair of types $xy \in \{ \alpha\beta, \alpha\gamma, \beta\gamma\}$ we will construct a new alphabet $\Sigma_{xy} = \{ \sigma_{xy} \mid \sigma \in \Sigma \}$ in which we mark each letter with the pair of types it should be participating in.
For a sequence $s \in (\Sigma \cup \Sigma)^*$ we use the notation $[s]_{xy}$ to represent the sequence in $(\Sigma_{xy} \cup \Sigma'_{xy})^*$ in which we replace every letter $\sigma$ with the letter $\sigma_{xy}$.

We will need three types of these clique gadgets in order to force the desired interaction between them.
\[
CG_\alpha(t) =  [ \mathtt{g} \ CNG(t) \  \mathtt{g} ]_{\alpha\gamma} \  \circ \ [ \mathtt{g}' \ p(CLG(t))^R \   \mathtt{g}' ]_{\alpha\beta}
\]
\[
CG_\beta(t) = [ \mathtt{g} \ CNG(t) \  \mathtt{g} ]_{\alpha\beta} \ \circ \ [ \mathtt{g}' \ p(CLG(t))^R \   \mathtt{g}' ]_{\beta\gamma}
\]
\[
CG_\gamma(t) =  [ \mathtt{g} \ CNG(t) \  \mathtt{g} \ ]_{\beta\gamma} \ \circ   \ [ \mathtt{g}' \ p(CLG(t))^R \   \mathtt{g}' ]_{\alpha\gamma}
\]

These clique gadgets achieve exactly what we want: for any three $k$-cliques $t_\alpha,t_\beta,t_\gamma \in \mathcal{C}_k$ the RNA folding score of the sequence $CG_\alpha(t_\alpha)\circ CG_\beta(t_\beta)\circ CG_\gamma(t_\gamma)$ is large (equal to some value $E_3$) if $t_\alpha \cup t_\beta \cup t_\gamma$ is a $3k$-clique and smaller otherwise (at most $E_3-1$).
In other words, an RNA folder can use these gadgets to determine if three separate $k$-cliques can form a $3k$-clique.
This is achieved by noticing that the highest priority for an optimal folding would be to match the $\mathtt{g}_{xy}$ letters with their counterparts $\mathtt{g}'_{xy}$, which leaves us with three of sequences to fold: $S_{\alpha\gamma} = [CNG(t_\alpha) \circ p(CLG(t_\gamma))^R]_{\alpha\gamma}$ and $S_{\alpha\beta} = [p(CLG(t_\alpha))^R \circ CNG(t_\beta)]_{\alpha\beta}$ and  $S_{\beta\gamma} = [p(CLG(t_\beta))^R \circ CNG(t_\gamma)]_{\beta\gamma}$.
The maximal score ($E_2$) in each one of these three sequences can be achieved iff every pair of our $k$-cliques form a $2k$-clique which happens iff they form a $3k$-clique.
This is formally proved in Claim~\ref{CG}.

The only remaining challenge is to combine all the sequences corresponding to all the $O(n^k)$ $k$-cliques in the graph into one sequence in a way that the existence of a ``good" triple, one that makes a $3k$-clique, affects the RNA folding score of the entire sequence.
Note that if we naively concatenate all the clique gadgets into one sequence, the optimal sequence will choose to fold clique gadgets in pairs instead of triples since folding a triple makes other gadgets unable to fold without crossings.
Instead, we will use the structure of the RNA folding problem again to implement a ``selection" gadget that forces exactly three clique gadgets to fold together in any optimal folding.
We remark that the implementation of such ``selection" gadgets is very different in the three proofs in this paper: In Section~\ref{sec:cfg} we use the derivation rules, in this section we use the fact that even when folding all expensive separators in a sequence to the left or right we are left with an interval that is free to fold with other parts of the sequence, and in Section~\ref{sec:dyck} we rely on the restriction of Dyck that an opening bracket can match only to closing brackets to its right.

Towards this end,  we introduce some extremely expensive symbols $\alpha,\beta,\gamma$.
Let $\ell_3=10 \ell_2$ be an upper bound on the total weight the $CG_x(t)$ gadgets, and set $w(\alpha)=w(\beta)=w(\gamma)=\ell_3$.
Our ``clique detecting" RNA sequence is defined as follows.

\begin{align*}
S_G \ =  & \ \ \ \  \mathtt{\alpha}^{2n^{k}}  \  \bigcirc_{t \in \mathcal{C}_k } \left(  \alpha' \ CG_\alpha(t) \ \alpha' \right) \ \alpha^{2n^{k}}  \\
& \circ \  \beta^{2n^{k}}  \  \bigcirc_{t \in \mathcal{C}_k } \left(  \beta' \ CG_\beta(t) \ \beta' \right) \ \beta^{2n^{k}}  \\
& \circ \ \gamma^{2n^{k}}  \  \bigcirc_{t \in \mathcal{C}_k } \left(  \gamma' \ CG_\gamma(t) \ \gamma' \right) \ \gamma^{2n^{k}}  
\end{align*} 

The added padding makes sure that all but one $CG_\alpha$ gadget are impossible to fold without giving up an extremely valuable $\alpha,\alpha'$ pair, and similarly all but one $CG_\beta$ and one $CG_\gamma$ cannot be folded.
To see this, assume all the $\alpha'$ are paired (left or right) and note that if both $\alpha'$ symbols surrounding a clique gadget $CG_\alpha(t)$ are paired to one side (say, left) then the only non crossing pairs that the gadget could participate in are either with $\alpha$ symbols (but those cannot be matches) or within itself.
Our marking of symbols with pairs of types $xy$ make it so that a clique gadget cannot have any matches with itself.
Therefore, if all $\alpha'$ symbols are matched, then all but one $CG_\alpha(t)$ gadgets do not participate in any foldings. The argument for $\beta,\gamma$ is symmetric.
We are left with a folding of a sequence of three clique gadgets $CG_\alpha(t_\alpha),CG_\beta(t_\beta),CG_\gamma(t_\gamma)$ which can achieve maximal score iff $t_\alpha \cup t_\beta \cup t_\gamma$ is a $3k$-clique.

This proves our main claim that the (weighted) RNA folding score of our clique detecting sequence $S_G$ is large (equal to some fixed value $E_C$) if the graph contains a $3k$-clique and smaller (at most $E_C-1$) otherwise.
See Claim~\ref{cl:main} for the formal proof.

Our final alphabet $\Sigma$ has size $18$ (together with $\Sigma'$ this makes $36$ symbols).
\[ \Sigma = \  \{ \alpha,\beta,\gamma \}  \ \cup \ \bigcup_{xy \in \{\alpha\beta,\alpha\gamma,\beta\gamma\}} \{ \mathtt{0, 1} ,\$, \#, g \}_{xy} 
\]

Observe that $S_G$ can be constructed from $G$ in $O(n^{k+1})$ time, by enumerating all subsets of $k$ nodes and that it has length $O(n^{k+1})$.
The construction of $S_G$ should be seen as a heavy preprocessing and encoding of the graph, after which we only have to work with $k$-cliques.
The largest weight we use in our construction is $\ell_3 = O(k^{O(1)} n\log{n})$ and therefore using Lemma~\ref{WRNA} we can reduce the computation of the weighted RNA of $S_G$ to an instance of (unweighted) RNA folding on a sequence of length $O(|S_G| k^{O(1)} n \log{n}) = \tilde{O}(k^{O(1)} n^{k+2})$ which proves Theorem~\ref{thm:RNA}.

\paragraph{Formal Proofs}
We will start with the proof that the list and node gadgets have the desired functionality.
Let $E_1= (10 n+1) \cdot \log n$.

\begin{claim} \label{nbrhood} 
For any $xy \in \{\alpha\beta, \alpha\gamma, \beta\gamma\}$ and two nodes $u,v \in V(G)$ we have that the weighted RNA folding score 
$ WRNA( [ NG(u) \circ p(LG(v))^R ]_{xy})$  is $E_1$ if $u \in N(v)$ and at most $E_1-1$ otherwise.
\end{claim}
\begin{proof}
	Since all letters in the sequence we are concerned with have the same mark $xy$, we will omit the subscripts.
	If $u \in N(v)$ then $p(\bar{u})^R$ appears in $LG(v)$ and we can completely match it with the sequence $\bar{u}$ in $NG(u)$, giving a score of $2\log{n}$, then we match all the $n$ $\$'$ symbols in $LG(v)$ to some $\$$ symbols in $NG(u)$ and gain an extra score of $n \cdot 10 \log n$.
	Therefore, in this case, the weighted RNA score is $E_1=\log n + 10 n \log n$.
	
	Now we assume that $u \notin N(v)$ and show that in the optimal folding the score is at most $E_1-1$.
	First, note that the sequence $p(LG(v))^R$ has fewer  $\$'$ symbols (it has $2n$ such symbols) than the sequence $NG(u)$ (which has $4n$ such symbols).
	 By not pairing a symbol $\$'$ in $p(LG(v))^R$, we lose a score of $w(\$)$ which is much more than the entire weight of the non-$\$$ symbols in $NG(u)$.
	 Therefore, any matching which leaves some $\$'$ unmatched is clearly sub-optimal, and we can assume	 that all the $\$'$ symbols are matched. Given this, the substring $\bar{u}$ can only be folding to at most one
	of the substrings  $p(\bar{v_i})^R$ in $LG(v)$ for some $v_i \in N(v)$. This folding can only achieve score $2\log{n}-1$ because $u \notin N(v)$.
	Thus, the total score of the optimal matching is no more than $E_1-1$.
	 \end{proof}

Next, we prove that the ``clique node gadgets" and ``clique list gadgets" check that two $k$-cliques form one bigger $2k$-clique.
Let $E_2 = 2k^2 \cdot \ell_1 + k^2 \cdot E_1$.

\begin{claim} \label{twoCG} 

For any $xy \in \{\alpha\beta, \alpha\gamma, \beta\gamma\}$ and two $k$-cliques $t_1,t_2 \in \mathcal{C}_k$ we have that the weighted RNA folding score 
$ WRNA( [ CNG(t_1) \circ p(CLG(t_2))^R ]_{xy})$  is $E_2$ if $t_1 \cup t_2$ is a $2k$-clique and at most $E_2-1$ otherwise.
\end{claim}
\begin{proof}
	We will omit the irrelevant $xy$ subscripts.
	First, note that the sequences $CNG(t_1)$ and $p(CLG(t_2))^R$ have the same number of $\#$ and $\#'$ symbols, respectively.
	By not pairing a single one of them with its counterpart, we lose a contribution of $w(\#)$ to the WRNA score,
	which is much more than we could gain by pairing all the symbols in all the node and list gadgets (that is, the rest of the sequence).
	Therefore, we assume that all the $\#$ and $\#'$ symbols are paired. 
	Let $t_1=\{u_1,\ldots,u_k\}$ and $t_2 = \{ v_1,\ldots,v_k\}$.
	We can now say that
	$$
		WRNA(CNG(t_1) \circ p(CLG(t_2))^R)=(2 k^2)w(\#)+\sum_{i \in [k]}\sum_{j \in [k]}WRNA(NG(u_i) \circ p(LG(v_j))^R).
	$$
	and by Claim~\ref{nbrhood} we know that $WRNA(NG(u_i) \circ p(LG(v_j))^R)=E_1$ if $u_i$ and $v_j$ are connected and less otherwise.
	Therefore, we can only get the maximal $E_2 = 2k^2\cdot \ell_1 + k^2 \cdot E_1$ iff every pair of nodes, one from $t_1$ and one from $t_2$ are connected.
	Since $u \notin N(u)$ for all $u \in V(G)$ and since $t_1,t_2$ are $k$-cliques, we conclude that $t_1 \cup t_2$ is a $2k$-clique.
\end{proof}

We are now ready to prove the main property of our clique gadgets: a sequence of three clique gadgets (one from each type) achieves maximal score iff they form a $3k$-clique together.
Let $E_3=6 \ell_2 + 3 E_2$.

\begin{claim} \label{CG} 

For any $xy \in \{\alpha\beta, \alpha\gamma, \beta\gamma\}$ and three $k$-cliques $t_\alpha,t_\beta,t_\gamma \in \mathcal{C}_k$ we have that the weighted RNA folding score 
$ WRNA( CG_\alpha(t_\alpha) \circ CG_\beta(t_\beta) \circ CG_\gamma(t_\gamma))$  is $E_3$ if $t_1 \cup t_2 \cup t_3$ is a $3k$-clique and at most $E_3-1$ otherwise.
\end{claim}

\begin{proof}
	If for some $xy \in \{\alpha\beta, \alpha\gamma, \beta\gamma\}$ there is a symbol $\mathtt{g}_{xy}$ which is not paired up with its counterpart, we lose
	a contribution to the $WRNA$ score that is more than we could get by pairing up all symbols that are not $g_{xy}$. Therefore, we have the equality
	\begin{align*}
		WRNA( CG_\alpha(t_\alpha) \circ CG_\beta(t_\beta) \circ CG_\gamma(t_\gamma)) = & \ 3 \cdot 2 \ell_2 \\
		& +WRNA( [ CNG(t_\alpha) \circ p(CLG(t_\gamma))^R ]_{\alpha\gamma}  ) \\
		& +WRNA( [ CNG(t_\beta) \circ p(CLG(t_\gamma))^R ]_{\beta\gamma}  ) \\
		& +WRNA( [  p(CLG(t_\alpha))^R \circ CNG(t_\beta) ]_{\alpha\beta}  )).
	\end{align*}
	By Claim~\ref{twoCG}, the last three summands are equal to $E_2$ if all our three $k$-cliques are pairwise $2k$-cliques and otherwise at least one of the summands is less than $E_2$.
	The claim follows by noticing that $t_\alpha \cup t_\beta \cup t_\gamma$ is a $3k$-clique iff the three $k$-cliques are pairwise $2k$-cliques. 
\end{proof}

We are now ready to prove our main claim about $S_G$.
This proof shows that our ``selection" gadgets achieve the desired property of having exactly one clique from each type fold in an optimal matching.
Let $N = O(n^{k})$ be the size of $\mathcal{C}_k$ which is the number of $k$-cliques in our graph and therefore the number of clique gadgets we will have from each type.
We will set $E_C=6N + E_3$.

\begin{claim} \label{cl:main} 
The weighted RNA score of $S_G$ is $E_C$ if $G$ contains a $3k$-clique and at most $E_C-1$ otherwise.
\end{claim}
\begin{proof}
	Let $x \in \{\alpha,\beta,\gamma\}$ and define $t_x \geq 0$ denote the number of  $x'$ symbols in $S_G$ that are not paired. 
	Because any clique gadget $CG_x$ can only have matches with letters from clique gadgets $CG_y$ for some $y \in \{\alpha,\beta,\gamma\}$ such that $y \neq x$, we can say that at most $t_x/2+1$ clique gadget sequences $CG_x$ can have letters that participate in the folding.

	Recall that by definition of our weights, the total weight of any clique gadget is much less than $\ell_3/10$ where $\ell_3$ is the weight of a letter $\alpha,\beta,\gamma$ and recall the definition of $N=| \mathcal{C}_k|$.
	We will use the inequalities:
	$$
		WRNA(S_G)\leq ((t_\alpha/2+1)+(t_\beta/2+1)+(t_\beta/2+1))\cdot \ell_3/10 + ((2 N -t_a)+(2N-t_b)+(2N-t_c))\ell_3,
	$$
	and:
		$$
		WRNA(S_G)\geq ((2 N -t_a)+(2N-t_b)+(2N-t_c))\ell_3.
	$$
	Since $\ell_3 \gg \ell_3/10$, we must have that $t_\alpha=t_\beta=t_\gamma=0$ in any optimal folding of $S_G$.
	Now we get that:
	$$
		WRNA(S_G)=6N \ell_3 + WRNA(CG_\alpha(t_\alpha) \circ CG_\beta(t_\beta) \circ CG_\gamma(t_\gamma))
	$$
	for some $k$-cliques $t_\alpha, t_\beta, t_\gamma \in \mathcal{C}_{k}$.
	By Claim~\ref{CG}, the last summand can be equal to $E_3$ 
	iff the graph has a $3k$-clique, and must be at most $E_3-1$ otherwise.
	This, and the fact that $E_C = 6N \ell_3 + E_3$ completes the proof.
\end{proof}

We are now ready to show that the construction of $S_G$ from  graph $G$ proves Theorem~\ref{thm:RNA}.

\begin{reminder}{Theorem~\ref{thm:RNA}}
If RNA Folding on a sequence of length $n$ can be solved in $T(n)$ time, then $k$-Clique on $n$ node graphs can be solved in $O\left(T \left(n^{k/3+O(1)} \right) \right)$ time, for any $k\geq 3$. 
Moreover, the reduction is combinatorial.
\end{reminder}

\begin{proof}
Given a graph $G$ on $n$ nodes we construct the sequence $S_G$ as described above.
The sequence can be constructed in $O(k^{O(1)}\cdot n^{k+1})$ time, by enumerating all subsets of $k$ nodes and that it has length $O(k^{O(1)}\cdot n^{k+1})$.
The largest weight we use in our construction is $\ell_3 = O(k^{O(1)} n\log{n})$ and therefore using Lemma~\ref{WRNA} we can reduce the computation of the weighted RNA of $S_G$ to an instance of (unweighted) RNA folding on a sequence of length $O(|S_G| k^{O(1)} n \log{n}) = \tilde{O}(k^{O(1)} n^{k+2})$.
Thus, an RNA folder as in the statement returns the weighted RNA folding score of $S_G$ in $T(n^{k/3+2})$ time (treating $k$ as a constant) and by Claim~\ref{cl:main} this score determines whether $G$ contains a $3k$-clique.
All the steps in our reduction are combinatorial.
\end{proof}

%% file: dyck_construction.tex
In this section we prove Theorem~\ref{thm:dyck} by reducing $k$-Clique to the Dyck Edit Distance problem, defined below.

The Dyck grammar is defined over a fixed size alphabet of opening brackets $\Sigma$ and of closing brackets $\Sigma' = \{ \sigma' \mid \sigma \in \Sigma \}$, such that $\sigma$ can only be closed by $\sigma'$.
A string $S$ belongs to the Dyck grammar if the brackets in it are well-formed.
More formally, the Dyck grammar is defined by the rules $\mathbf{S} \to \mathbf{S}\mathbf{S}$ and $\mathbf{S} \to \sigma \ \mathbf{S} \ \sigma'$ for all $\sigma \in \Sigma$ and $\mathbf{S} \to \eps$.
This grammar defines the Dyck context free language (which can be parsed in linear time).

The Dyck Edit Distance problem is: given a string $S$ over $\Sigma \cup \Sigma'$ find the minimum edit distance from $S$ to a string in the Dyck CFL.
In other words, find the shortest sequence of substitutions and deletions that is needed to convert $S$ into a string that belongs to Dyck.
We will refer to this distance as the Dyck \emph{score} or \emph{cost} of $S$.

Let us introduce alternative ways to look at the Dyck Edit Distance problem that will be useful for our proofs.
Two pairs of indices $(i_1,j_1),(i_2,j_2)$ such that $i_1<j_1$ and $i_2<j_2$ are said to ``cross" iff at least one of the following three conditions hold
\begin{itemize}
	\item $i_1=i_2$ or $i_1=j_2$, or $j_1=i_2$, or $j_1=j_2$;
	\item $i_1<i_2<j_1<j_2$;
	\item $i_2<i_1<j_2<j_1$.
\end{itemize}
Note that by our definition, non-crossing pairs cannot share any indices. 
We define an \emph{alignment} $A$ of a sequence $S$ of length $n$ to be a set of non-crossing pairs $(i,j), i<j, i,j \in [n]$.
If $(i,j)$ is in our alignment we say that letter $i$ and letter $j$ are \emph{aligned}.
We say that an aligned pair is a \emph{match} if $S[i] = \sigma$ for some $\sigma \in \Sigma$ and $S[j]=\sigma'$, i.e. an opening bracket and the corresponding closing bracket.
Otherwise, we say that the aligned pair is a \emph{mismatch}. Mismatches will correspond to substitutions in an edit distance transcript.
A letter at an index $i$ that does not appear in any of the pairs in the alignment is said to be \emph{deleted}.
We define the cost of an alignment to be the number of mismatches plus the number of deleted letters.
One can verify that any alignment of cost $E$ corresponds to an edit distance transcript from $S$ to a string in Dyck of cost $E$, and vice versa.

\subsection{The Reduction}

Given a graph $G=(V,E)$ on $n$ nodes and $O(n^2)$ unweighted undirected edges, we will describe how to efficiently construct a sequence $S_G$ over an alphabet $\Sigma$ of constant size, such that the Dyck score of $S_G$ will depend on whether $G$ contains a $3k$-clique.
The length of $S_G$ will be $O(k^d n^{k+c})$ for some small fixed constants $c,d>0$ independent of $n$ and $k$, and the time to construct it from $G$ will be linear in its length.
This will prove that a fast (e.g. subcubic) algorithm for Dyck Edit Distance can be used as a fast $3k$-clique detector (one that runs much faster than in $O(n^{3k})$ time).

As in the other sections, our main strategy will be to enumerate all $k$-cliques in the graph and then search for a triple of $k$-cliques that have all the edges between them.
We will be able to find such a triple iff the graph contains a $3k$-clique.
A Dyck Edit Distance algorithm will be utilized to speed up the search for such a ``good" triple.
Our reduction will encode every $k$-clique of $G$ using a ``short" sequence of length $O(n^c)$ such that the Dyck score of a sequence composed of the encodings of a triple of sequences will be large iff the triple is ``good".
Then, we will show how to combine the short encodings into our long sequence $S_G$ such that the existence of a ``good" triple affects the overall score of an optimal alignment.

\paragraph{The Sequence}
Our sequence $S_G$ will be composed of many smaller gadgets which will be combined in certain ways by other padding gadgets.
We construct these gadgets now and explain their useful properties.
The proofs of these properties are postponed until after we present the whole construction of $S_G$.

Recall that we associate every node in $V(G)$ with an integer in $[n]$ and let $\bar{v}$ denote the encoding of $v$ in binary and we will assume that it has length exactly $2\log{n}$ for all nodes.
We will use the fact that there is no node with encoding $\bar{0}$.
For a sequence $s \in \Sigma^*$ let $p(s) \in (\Sigma')^*$ be the sequence obtained from $s$ by replacing every letter $\sigma \in \Sigma$ with the closing bracket $\sigma' \in \Sigma'$.
That is, if $s=s_1\cdots s_n$ then $p(s)=s_1'\cdots s_n'$.

 Our alphabet $\Sigma$ will contain the letters $\mathtt{0,1}$ and some additional symbols which we will add as needed in our gadgets like $\$,\#$.
 We will use the numbers $\ell_0,\ldots,\ell_5$ such that $\ell_i = (1000\cdot n^{2})^{i+1}$, which can be bounded by $n^{O(1)}$.
  We define \emph{node} gadgets as, 
 \[ 
 NG(v) = \$^{\ell_1} \  \bar{v} \  \$^{\ell_1}
 \]
 and \emph{list} or \emph{neighborhood} gadgets as,
 \[
 LG(v) =   \bigcirc_{u \in N(v)} (\$^{\ell_0} \ \bar{u} \ \$^{\ell_0}) \ \circ \ \bigcirc_{u \notin N(v)} (\$^{\ell_0} \ \bar{0}\ \$^{\ell_0}).
 \]
 
These gadget are constructed so that for any two nodes $u,v \in V(G)$, the Dyck score of the sequence $NG(v) \circ p\left(LG(u) \right)^R$ is small (equal to some fixed value $E_1$) if $v$ is in the neighborhood of $u$, that is $(u,v) \in E(G)$, and larger otherwise (at least $E_1+1$).
This is proved formally in Claim~\ref{claim:neighb}, by similar arguments as in Section~\ref{sec:RNA}.

Note that $\ell_2$ is an upper bound on the total length of all the symbols in the gadgets $NG(v)$ and $LG(v)$, for any node $v \in V(G)$.
Let $\mathcal{C}_k$ be the set of $k$-cliques in $G$ and consider some $t = \{ v_1,\ldots,v_k\} \in \mathcal{C}_k$.
We will now combine the node and list gadgets into larger gadgets that will be encoding $k$-cliques.
%
We will encode a clique in two ways. The first one is,
\[
CNG(t) =  \bigcirc_{v \in t} \  (\#^{\ell_2}  \ NG(v) \ \#^{\ell_2})^k 
\]
and the second one is,
\[
CLG(t) =  (\bigcirc_{v \in t} \  (\#^{\ell_2} \ LG(v) \  \#^{\ell_2} ))^k .
\]

Note that $\ell_3 $ is an upper bound on the total length of all the symbols in the $CNG(t)$ and $CLG(t)$ gadgets.
We will add the symbol $\mathtt{g}$ to the alphabet.
Moreover, we will now duplicate our alphabet three times  to force only ``meaningful" alignments between our gadgets.
It will be convenient to think of $\alpha,\beta,\gamma$ as three \emph{types} such that we will be looking for three $k$-cliques, one from type $\alpha$ one from $\beta$ and one from $\gamma$.
For any  pair of types $xy \in \{ \alpha\beta, \alpha\gamma, \beta\gamma\}$ we will construct a new alphabet $\Sigma_{xy} = \{ \sigma_{xy} \mid \sigma \in \Sigma \}$ in which we mark each letter with the pair of types it should be participating in.
For a sequence $s \in (\Sigma \cup \Sigma')^*$ we use the notation $[s]_{xy}$ to represent the sequence in $(\Sigma_{xy} \cup \Sigma'_{xy})^*$ in which we replace every letter $\sigma$ with the letter $\sigma_{xy}$.

We will need three types of these clique gadgets in order to force the desired interaction between them.
\begin{align*}
CG_\alpha(t) \ \  =  \ \  \mathtt{a^{\ell_4} \ (x_\alpha')^{\ell_5}} \ \  &[ \mathtt{g}^{\ell_3} \ CNG(t) \  \mathtt{g}^{\ell_3} ]_{\alpha\gamma} \  & &  \circ \  & & [ \mathtt{g}^{\ell_3} \ p(CNG(t))^R \   \mathtt{g}^{\ell_3} ]_{\alpha\beta} \ \ \mathtt{y_\alpha^{\ell_5} \ (a')^{\ell_4}}\\
CG_\beta(t)  \ \ =  \ \ \mathtt{b^{\ell_4} \ (x_\beta')^{\ell_5}} \ \  &[ (\mathtt{g}')^{\ell_3} \ CLG(t) \  (\mathtt{g}')^{\ell_3} ]_{\alpha\beta} \  & & \circ \  & & [ \mathtt{g}^{\ell_3} \ p(CNG(t))^R \   \mathtt{g}^{\ell_3} ]_{\beta\gamma} \ \ \mathtt{y_\beta^{\ell_5} \ (b')^{\ell_4}}\\
CG_\gamma(t)  \ \ = \ \  \mathtt{c^{\ell_4} \ (x_\gamma')^{\ell_5}}  \ \ &[ (\mathtt{g}')^{\ell_3} \ CLG(t) \  (\mathtt{g}')^{\ell_3} \ ]_{\beta\gamma} \ & &  \circ   \  & & [ (\mathtt{g}')^{\ell_3} \ p(CLG(t))^R \   (\mathtt{g}')^{\ell_3} ]_{\alpha\gamma} \ \ \mathtt{y_\gamma^{\ell_5} \ (c')^{\ell_4}}
\end{align*}

These clique gadgets achieve exactly what we want: for any three $k$-cliques $t_\alpha,t_\beta,t_\gamma \in \mathcal{C}_k$ the Dyck score of the sequence $CG_\alpha(t_\alpha)\circ CG_\beta(t_\beta)\circ CG_\gamma(t_\gamma)$ is small (equal to some value $E_3$) if $t_\alpha \cup t_\beta \cup t_\gamma$ is a $3k$-clique and larger otherwise (at least $E_3+1$).
This is formally proved in Claim~\ref{claim:clique}, again by similar arguments as in Section~\ref{sec:RNA} (but more complicated because of the possible mismatches).

The main difference over the proof of Section~\ref{sec:RNA} is the way we implement the ``selection" gadgets.
We want to combine all the clique gadgets into one sequence in a way that the existence of a ``good" triple, one that makes a $3k$-clique, affects the Dyck score of the entire sequence.
The ideas we used in the RNA proof do not immediately work here because of ``beneficial mismatches" of the separators we add with themselves and because in Dyck $(\sigma,\sigma')$ match but $(\sigma',\sigma)$ do not (while in RNA we do not care about the order). 
We will use some new ideas.

%
Our ``clique detecting" sequence is defined as follows.

\begin{align*}
S_G \ =  & \ \ \ \  \mathtt{x_\alpha}^{\ell_5}  \  \left(\bigcirc_{t \in \mathcal{C}_k }  CG_\alpha(t)  \right) \ \mathtt{y_\alpha'}^{\ell_5}  \\
& \circ \  \mathtt{x_\beta}^{\ell_5}  \  \left(\bigcirc_{t \in \mathcal{C}_k }   CG_\beta(t)  \right) \ \mathtt{y_\beta'}^{\ell_5}  \\
& \circ \  \mathtt{x_\gamma}^{\ell_5}  \  \left(\bigcirc_{t \in \mathcal{C}_k }    CG_\gamma(t)  \right) \ \mathtt{y_\gamma'}^{\ell_5} 
\end{align*} 

As the $\mathtt{x_\alpha,y'_\alpha}$ symbols are very rare and	 ``expensive" an optimal alignment will match them to some of their counterparts within the $\alpha$ part of the sequence.
However, when the $\mathtt{x'_\alpha,y_\alpha}$ letters are matched, we cannot match the adjacent $\mathtt{a,a'}$ symbols - which are also quite expensive.
Therefore, the optimal behavior is to match the $\mathtt{x_\alpha,y'_\alpha}$ from \emph{exactly one} ``interval" from the $\alpha$ part.
A similar argument holds for the $\beta,\gamma$ parts.
This behavior leaves exactly one clique gadget from each type to be aligned freely with each other as a triple.
By the construction of these gadgets, an optimal score can be achieved iff there is a $3k$-clique.

%
This proves our main claim that the Dyck score of our clique detecting sequence $S_G$ is small (equal to some fixed value $E_C$) if the graph contains a $3k$-clique and larger (at least $E_C+1$) otherwise.
See Claim~\ref{claim:main} for the formal proof.

When $k$ is fixed, $S_G$ can be constructed from $G$ in $O(n^{k+O(1)})$ time, by enumerating all subsets of $k$ nodes and that it has length $O(n^{k+O(1)})$.
This proves Theorem~\ref{thm:dyck}.

Our final alphabet $\Sigma$ has size $24$ (together with $\Sigma'$ this makes $48$ symbols).
\[ \Sigma = \  \{ \mathtt{a,b,c} \}  \ \cup \ \bigcup_{xy \in \{\alpha\beta,\alpha\gamma,\beta\gamma\}} \{ \mathtt{0, 1} ,\$, \#, \mathtt{g,x,y} \}_{xy} 
\]

\paragraph{Formal Proofs}

%% file: dyck_part1.tex
Let $E_1=\log{n}\cdot(n-1)+(2\ell_1-n\cdot 2\ell_0)/2$.

\begin{claim} \label{claim:neighb}
	For any $xy \in \{\alpha\beta,\alpha\gamma,\beta\gamma\}$, if $v \in N(u)$, then
	$$
		Dyck([NG(v)\circ p(LG(u))]_{xy})=E_1
	$$
	and $>E_1$ otherwise.
\end{claim}
\begin{proof}
We will omit the subscripts $xy$ since they do not matter for the proof.
	We want to claim that the binary sequence $\bar{v}$ is aligned to at most one binary sequence $p(\bar{z}^R)$. If this is not so, then there are $2\ell_0$ symbols $\$'$ from $p(LG(u))$ that will be mismatched or deleted thus contributing at least $S_1=\ell_0$ to the $Dyck$ score. There will be at least $2\ell_1-(n-1)2\ell_0$ symbols $\$$ that are not matched to their counterparts $\$'$. Those will contribute at least $S_2=\ell_1-(n-1)\ell_0$ to the $Dyck$ score. We get that the $Dyck$ score is at least $S_1+S_2>E_1$. Now we assume that the binary sequence $\bar{v}$ is aligned to exactly one other binary sequence which we denote by $p(\bar{z_1}^R)$. There are $S_2=2\log(n)\cdot (n-1)$ symbols from binary sequences from $p(LG (u))$ that are not matched to their counterparts. Also, there are $S_3=2\ell_1-n\cdot 2\ell_0$ symbols $\$$ that are not matched to their counterparts. The contribution of the unmatched symbols is $\geq (S_2+S_3)/2=E_1$ to the $Dyck$ cost. We want to show that the equality can be achieved iff $v \in N(u)$. From the proof it is clear that, if $v \in N(u)$, then we can achieve equality by choosing $z_1$ to be the element from $N(u)$ that is equal to $v$. Also, if we achieve equality, the symbols corresponding to $S_2$ and $S_3$ contribute $\geq (S_2+S_3)/2$ to the $Dyck$ score. These symbols contribute $(S_2+S_3)/2$ to the $Dyck$ cost iff the mismatches happen only between  themselves and all symbols $\$'$ are matched to their counterparts. The only remaining symbols that could potentially contribute to the score correspond to $v$ and $z_1$. They contribute $0$ to the $Dyck$ score iff $\bar{v}=\bar{z_1}^R$, that is, $v \in N(u)$.
\end{proof}

For the next proofs we will use the following definition.
\begin{definition}
	Given two sequences $P$ and $T$, we define
	$$
		pattern(P,T):=\min_{\substack{Q\text{ is a contiguous}\\ \text{subsequence of }T}}Dyck(P\circ Q).
	$$
\end{definition}

Let $E_2=k^2\cdot E_1$.
\begin{claim} \label{claim:biclique}
	For any $xy \in \{\alpha\beta,\alpha\gamma,\beta\gamma\}$ and two $k$-cliques $t_1,t_2 \in \mathcal{C}_k$, if $t_1 \cup t_2$ is a $2k$ clique then
	$$
		Dyck([CNG(t_1)\circ CLG(t_2)]_{xy})= E_2
	$$
	and $> E_2$ otherwise.
\end{claim}
\begin{proof}
We will omit the subscripts $xy$ since they do not matter for the proof.
	We have that
	$$
		Dyck(CNG (t_1)\circ CLG (t_2))
		\geq\sum_{v \in t_1} k \cdot pattern(NG (v),CLG (t_2)).
	$$
	Suppose that for some $v \in t_1$, $NG (v)$ is aligned to more than one gadget $p(LG (u))$.
	Then symbols $\#$ between these gadgets $p(LG (u))$ will be substituted or deleted. The cost of these
	operations is $\geq \ell_2> E_2$. Therefore, we have that any one of $k^2$ gadgets $NG (v)$ is aligned to at most one gadget
	$p(LG (u))$ for some $u \in t_2$. By the construction of $CNG$ and $CLG$, we have that
	\begin{align*}
		&Dyck(CNG (t_1)\circ CLG (t_2))\\
		\geq &\sum_{v \in t_1} \sum_{u \in t_2} pattern(NG (v),(\#')^{2\ell_2} p(LG (u)) (\#')^{2\ell_2})\\
		=&\sum_{v \in t_1} \sum_{u \in t_2} Dyck(NG (v)\circ p(LG (u))),
	\end{align*}
	where the last equality follows because $\#$ does not appear among symbols of $NG (v)$.
	Now we have that
	$$
		Dyck(CNG (t_1)\circ CLG (t_2)) \geq \sum_{v \in t_1} \sum_{u \in t_2} E_1=k^2\cdot E_1= E_2,
	$$
	where we use Claim \ref{claim:neighb}.
	If we have equality, it means that we have equality in all $k^2$ invocations of Claim \ref{claim:neighb}, which implies that $v \in N(u)$ for all $v \in t_1$, $u \in t_2$. And this gives that there is a biclique between vertices of $t_1$ and $t_2$. Also, it is possible to verify that we can achieve the equality if there is a biclique.
\end{proof}

Let $E_3=3({\ell_4}+ E_2)$. 

\begin{claim} \label{claim:clique}
	For any triple of $k$-cliques $t_\alpha,t_\beta,t_\gamma \in \mathcal{C}_k$, the union $t_\alpha \cup t_\beta \cup t_\gamma$ is a $3k$-clique, then
	\begin{align*}
		Dyck(
			& \ \ \ \mathtt{x_\alpha}^{{\ell_5}} CG_\alpha(t_\alpha) (\mathtt{y_\alpha}')^{{\ell_5}} \\
			&\circ \ \mathtt{x_\beta}^{{\ell_5}} CG_\beta(t_\beta) (\mathtt{y_\beta}')^{{\ell_5}} \\
			&\circ \ \mathtt{x_\gamma}^{{\ell_5}} CG_\gamma(t_\gamma) (\mathtt{y_\gamma}')^{{\ell_5}}
			\ \ \ )=E_3
	\end{align*}
	and $>E_3$ otherwise.
\end{claim}
\begin{proof}
	We need to lower bound
	\begin{align*}
		Dyck( \ \ & \mathtt{x_\alpha}^{{\ell_5}} \mathtt{a}^{\ell_4} (\mathtt{x_\alpha}')^{{\ell_5}} \mathtt{g}_{\alpha\gamma}^{\ell_3} & & CNG_{\alpha\gamma}(t_\alpha) \mathtt{g}_{\alpha\gamma}^{\ell_3} \mathtt{g}_{\alpha\beta}^{\ell_3} & & CNG_{\alpha\beta}(t_\alpha) \mathtt{g}_{\alpha\beta}^{\ell_3} \mathtt{y_\alpha}^{{\ell_5}} (\mathtt{a'})^{\ell_4} (\mathtt{y_\alpha}')^{{\ell_5}} \\
			 & \mathtt{x_\beta}^{{\ell_5}} \mathtt{b}^{\ell_4} (\mathtt{x_\beta}')^{{\ell_5}}( \mathtt{g}_{\alpha\beta}')^{\ell_3} & &  CLG_{\alpha\beta}(t_\beta) (\mathtt{g}_{\alpha\beta}')^{\ell_3} \mathtt{g}_{\beta\gamma}^{\ell_3} & & CNG_{\beta\gamma}(t_\beta) \mathtt{g}_{\beta\gamma}^{\ell_3} \mathtt{y_\beta}^{{\ell_5}} (\mathtt{b'})^{\ell_4} (\mathtt{y_\beta}')^{{\ell_5}} \\
			 & \mathtt{x_\gamma}^{{\ell_5}} \mathtt{c}^{\ell_4} (\mathtt{x_\gamma}')^{{\ell_5}}( \mathtt{g}_{\beta\gamma}')^{\ell_3} & & CLG_{\beta\gamma}(t_\gamma) (\mathtt{g}_{\beta\gamma}')^{\ell_3}( \mathtt{g}_{\alpha\gamma}')^{\ell_3} & & CLG_{\alpha\gamma}(t_\gamma) (\mathtt{g}_{\alpha\gamma}')^{\ell_3} \mathtt{y_\gamma}^{{\ell_5}} (\mathtt{c'})^{\ell_4} (\mathtt{y_\gamma}')^{{\ell_5}} \ \ ).
	\end{align*}
	Assume that some symbol $\mathtt{a}$ is aligned to a symbol to the right of $\mathtt{x_\alpha}'$. Then sequence $(\mathtt{x_\alpha}')^{{\ell_5}}$ will contribute $\geq {\ell_5}/2>E_3$ to the Dyck score and we are done. (We prove later that we can achieve Dyck score $E_3$ if there is a clique.)
	Now let $r$ denote the number of symbols from sequence $j=\mathtt{x_\alpha}^{{\ell_5}} \mathtt{a}^{\ell_4} (\mathtt{x_\alpha}')^{{\ell_5}}$ that are aligned to a symbol that does not belong to sequence $j$. Let $s$ denote the set of these $r$ symbols. Let $l$ denote symbols $\mathtt{x_\alpha}$ that are aligned to a symbol $\mathtt{x_\alpha}'$. There are $2{{\ell_5}}+{\ell_4}-r-2l$ symbols from $j$ that are not considered yet. These symbols are not matched to their counterparts and, therefore, contribute at least $\lceil(2{{\ell_5}}+{\ell_4}-r-2l)/2\rceil$ to the Dyck score (we divide by $2$ because the symbols can be mismatched among themselves in pairs). We have that $\lceil(2{{\ell_5}}+{\ell_4}-r-2l)/2\rceil+r\geq {\ell_4}/2+\lceil r/2 \rceil$ by the definition of $l$ (it implies that $l\leq {{\ell_5}}$). We note that ${\ell_4}/2=Dyck(j)$. $Dyck(j)\leq {\ell_4}/2$ can be obtained by aligning symbols $\mathtt{x_\alpha}$ with symbols $\mathtt{x_\alpha}'$ and mismatching symbols $\mathtt{a}$ in pairs. The reverse inequality follows by observing that all symbols $\mathtt{a}$ will be mismatched. Also we note that, if we mismatch symbols from $s$ among themselves, this costs $\lceil r/2 \rceil$. All this gives that we can assume that symbols in $j$ do not interact with symbols that are not in $j$ when we want to bound the Dyck score. Similarly, we can argue when $j=\mathtt{y_\alpha}^{{\ell_5}} (\mathtt{a}')^{\ell_4} (\mathtt{y_\alpha}')^{{\ell_5}}, \mathtt{x_\gamma}^{{\ell_5}} \mathtt{b}^{\ell_4} (\mathtt{x_\gamma}')^{{\ell_5}}, \mathtt{y_\beta}^T (\mathtt{b}')^{\ell_4} (\mathtt{y_\beta}')^{{\ell_5}}, \mathtt{x_\gamma}^{{\ell_5}} \mathtt{c}^{\ell_4} (\mathtt{x_\gamma}')^{{\ell_5}}, \mathtt{y_\gamma}^{{\ell_5}} (\mathtt{c}')^{\ell_4} (\mathtt{y_\gamma}')^{{\ell_5}}$. Thus, we need to show that 
	\begin{align*}
		Dyck( \ \ \ & \mathtt{g}_{\alpha\gamma}^{\ell_3}  & & CNG_{\alpha\gamma}(t_\alpha) & &  \mathtt{g}_{\alpha\gamma}^{\ell_3}  & & \mathtt{g}_{\alpha\beta}^{\ell_3}  & & CNG_{\alpha\beta}(t_\alpha) & &  \mathtt{g}_{\alpha\beta}^{\ell_3} \\
			 & (\mathtt{g}_{\alpha\beta}')^{\ell_3}  & & CLG_{\alpha\beta}(t_\beta)  & & (\mathtt{g}_{\alpha\beta}')^{\ell_3} & &  \mathtt{g}_{\beta\gamma}^{\ell_3}  & & CNG_{\beta\gamma}(t_\beta) & &  \mathtt{g}_{\beta\gamma}^{\ell_3} \\
			 & (\mathtt{g}_{\beta\gamma}')^{\ell_3}  & & CLG_{\beta\gamma}(t_\gamma) & &  (\mathtt{g}_{\beta\gamma}')^{\ell_3} & &  (\mathtt{g}_{\alpha\gamma}')^{\ell_3}  & & CLG_{\alpha\gamma}(t_\gamma) & &  (\mathtt{g}_{\alpha\gamma}')^{\ell_3} \ \ \ )\geq 3 E_2
	\end{align*}
	with equality iff there is a biclique between vertices of $t_\alpha$ and $t_\beta$,
	between vertices of $t_\alpha$ and $t_\gamma$ and between vertices of $t_\beta$ and $t_\gamma$. Let $h$ be the argument to Dyck function, that is, we want to show that $Dyck(h)\geq 3 E_2$ with the stated condition for the equality.

	Consider three gadgets $CNG$ and three gadgets $CLG$ as above. 
	We can assume that no symbol of any of these six gadgets is aligned to any symbol $\mathtt{g}_{xy}$ or $\mathtt{g}_{xy}'$. Assume that it is not the case. Then we can delete all symbols from the gadgets that are aligned to symbols $\mathtt{g}_{xy}$ or $\mathtt{g}_{xy}'$. After this, we rematch $\mathtt{g}_{xy}$ or $\mathtt{g}_{xy}'$ among themselves. We can check that we can always make this rematching of symbols $\mathtt{g}_{xy}$ or $\mathtt{g}_{xy}'$ so that the cost do not increase.
	Furthermore, if some $CNG_{xy}$ gadget is aligned with a $CNG_{x'y'}$ (or $CLG_{x'y'}$) for $(x,y)\neq (x',y')$, then there are two substrings of the type $\mathtt{g}_{ab}$ or $\mathtt{g}_{ab}'$ that don't have their counterpart between $CNG_{xy}$ and $CNG_{x'y'}$ (or $CLG_{x'y'}$). Hence, their contribution to the Dyck score is at least $2\ell_3/2={\ell_3}\gg 3 E_2$. Thus for all $(x,y)$ the only gadget that $CNG_{xy}$ can be aligned with is $CLG_{xy}$ and vice versa.
	This means that we can assume that all the $g$ and $g'$ symbols are completely aligned.
	
We have shown that the Dyck cost of the string is exactly 

	\begin{align*}
		& Dyck(CNG_{\alpha\gamma}(t_\alpha)\circ CLG_{\alpha\gamma}(t_\alpha)) \\
		+ & Dyck(CNG_{\alpha\beta}(t_\beta)\circ CLG_{\alpha\beta}(t_\beta)) \\
		+ & Dyck(CNG_{\beta\gamma}(t_\gamma)\circ CLG_{\beta\gamma}(t_\gamma)). 
	\end{align*}

We want to show that this is $\geq 3E_2$ with equality iff $t_\alpha\cup t_\beta\cup t_\gamma$ is a $3k$-clique.
This was shown in Claim \ref{claim:biclique}. 
\end{proof}

We now turn to the proof of the claim about the behavior of our ``selection gadgets".
Let $E_C$ be a fixed integer to be defined later that depends on $n$, $N$, and $k$.

\begin{claim}
\label{claim:main}
	If the $G$ contains a $3k$-clique, then $Dyck(S_G)=E_C$ and $Dyck(S_G)>E_C$ otherwise.
\end{claim}

The proof of this claim will require several claims and lemmas.
We start with some lemmas about general properties of Dyck Edit Distance. 

	\begin{lemma} \label{allmismatches}
		Let $Z_1$ be a substring of sequence $Z$. Assume that $Z_1$ is of even length. 
		If all symbols symbols from $Z_1$ participate in mismatches and deletions only, then we can modify the alignment so that there is no symbol in $Z_1$ that is aligned to a symbol that is not in $Z_1$.
	\end{lemma}
	\begin{proof}
		Let $l$ denote number of symbols from $Z_1$ that are aligned to symbols that are outside of $Z_1$. Let $s$ denote the set of all these symbols outside of $Z_1$ that are aligned to $Z_1$. There are two cases to consider.
			\begin{itemize}
				\item 
					$l$ is even. Then the $Dyck$ cost induced by symbols in $s$ and in $Z_1$ is at least $S_1:=l+(|Z_1|-l)/2$ by the properties from the statement of the lemma. We do the following modification to the alignment. We align all symbols in $s$ among themselves in pairs. This induces cost $S_2:=l/2$. We align all symbols in $Z_1$ among themselves. This induces cost $S_3=|Z_1|/2$. The total induced cost after the modification is $S_2+S_3\leq S_1$ and we satisfy the requirement in the lemma.
				\item
					$l$ is odd. Then the $Dyck$ cost induced by symbols in $s$ and in $Z_1$ is at least $S_1:=l+(|Z_1|-l+1)/2$. We do the following modification to the alignment. We align all symbols in $s$ among themselves in pairs except one symbol, which we delete (remember that $l$ is odd). This induces cost $S_2:=(l+1)/2$. We align all symbols in $Z_1$ among themselves. This induces cost $S_3=|Z_1|/2$. The total induced cost after the modification is $S_2+S_3\leq S_1$ and we satisfy the requirement in the lemma.
			\end{itemize}
	\end{proof}
	
	Consider optimal alignment of some string $w \in (\Sigma \cup \Sigma')^*$.
	\begin{lemma} \label{rightleft}
		Let $Z$ be a maximal substring of $w$ consisting entirely of symbols $z$ (for some symbol $z$ appearing in $w$).
		Let $Z_0$ be a maximal substring of $w$ consisting entirely of symbols $z_0$.
		Let $|Z|=|Z_0|$ and let there are matches between $Z$ and $Z_0$.
		Also, there are no other maximal substrings of $w$ containing $z$ other than $Z$.
		Then we can increase the $Dyck$ score by at most $2$ by modifying the alignment and get that the symbols of $Z_0$,
		that are matched to $z$, form a substring of $Z_0$ and the substring is suffix or prefix of $Z_0$ (we can choose whether it is a suffix of a prefix). We can also assume that the rest of symbols in $Z_0$ are deleted or mismatched among themselves.
	\end{lemma}
	\begin{proof}
		Wlog, we will show that we can make the substring to be the suffix.
		We write $Z_0=Z_1 Z_2 Z_3$ so that the first symbol of $Z_2$ is the first symbol of $Z_0$ that is aligned to some symbol $z$ and the last symbol of $Z_2$ is the last symbol of $Z_0$ that is aligned to some symbol $z$. First, we modify the alignment as follows. If the length of $Z_1$ is even, we don't do anything. Otherwise, consider the first symbol of $Z_1$. If it is aligned to some symbol, delete the first symbol of $Z_1$ and the symbol aligned to it by increasing the $Dyck$ cost by $1$. Similarly delete the last symbol of $Z_3$ and the symbol aligned to it if $Z_3$ is of odd length. Now, by Lemma \ref{allmismatches}, we can assume that all symbols in $Z_1$ and $Z_3$ are mismatched among themselves (except, possibly, the first symbol of $Z_1$ and the last symbol of $Z_3$). Now we are at the state when all symbols from $Z_0$ are mismatched among themselves except few that are matched with symbols $z$ from $Z$. Now we can rematch these symbols from $Z$ with symbols $z_0$ from $Z_0$ so that $z_0$ come from suffix of $Z_0$. We see that this do not increase the $Dyck$ score besides two possible deletions.
	\end{proof}

	Now we prove some claims about the properties of the optimal alignment of $S_G$.
	These claims essentially show that any ``bad behaviour", in which we do not align exactly one clique gadget of each type, is suboptimal.
	
	Let  $E_l$ be some value that depends only on $n$ and $k$.
	\begin{claim} \label{close}
		For any gadget $CG_\beta(t)$ from the sequence,
		if none of the $\mathtt{x}_\beta', \mathtt{y}_\beta$ symbols are matched to their counterparts, then
		all the $\mathtt{b}$ symbols from $CG_\beta(t)$ must be matched to their counterparts $\mathtt{b}'$ from $CG_\beta(t)$, and in this case, the gadget contributes $E_l$ to the $Dyck$ score.
		Analogous claims hold for $\alpha,\gamma$.
	\end{claim}
	\begin{proof}
		By Lemma \ref{allmismatches} we can assume that all symbols $\mathtt{x}_\beta'$ are mismatched among themselves. The same we can say about symbols $\mathtt{y}_\beta$. Let $Z$ be the substring of the gadget between symbols $\mathtt{x}_\beta'$ and $\mathtt{y}_\beta$. If $Z$ has matches with some symbols, then, by the construction of $w$, no symbol $\mathtt{b}$ or $\mathtt{b}'$ is matched to its counterpart and by Lemma \ref{allmismatches}, we get that $\mathtt{b}$ and $\mathtt{b}'$ are mismatched within themselves. But now we can decrease the $Dyck$ score by deleting all symbols from $Z$ and all symbols that $Z$ is aligned to outside the gadget. This increases the $Dyck$ score but then we can decrease it substantially by matching all symbols $\mathtt{b}$ to their counterparts $\mathtt{b}'$ from the gadget. In the end we get smaller $Dyck$ score because $\ell_4\geq 100 |Z|$. Now it remains to consider the case when symbols in $Z$ do not participate in matches. But then by Lemma~\ref{allmismatches}, we again conclude that symbols in $Z$ participate only in mismatches and only among themselves. 
		Let $s$ be the union of all symbols $\mathtt{b}$ and $\mathtt{b}'$ of the gadget and all symbols that these symbols are aligned to outside the gadget. Let $l$ denote the number of symbols from $s$ that are not coming from the gadget. Consider two cases.
		\begin{itemize}
			\item $l=0$. We satisfy requirements of the claim by matching all symbols $\mathtt{b}$ to their counterparts $\mathtt{b}'$.
			\item There is no symbol among the $l\geq 1$ symbols that participate in a match. We have that all symbols in $s$ contribute at least $l$ to the $Dyck$ score. We modify the alignment as follows. We match all symbols $\mathtt{b}$ to their counterparts $\mathtt{b}'$. We match the rest $l$ symbols from $s$ among themselves. If there is odd number of them, we delete one. This contributes at most $S_1:=(l+1)/2$ to the Dyck score after the modification. We have that $S_1\leq l$ for $l\geq 1$.
			\item Complement of the previous two cases: there is a symbol among the $l \geq 1$ symbols that participates in a match.
			Wlog, symbols $\mathtt{b}$ from the gadget participate in at least one matching. Then all symbols $\mathtt{b}'$ from the gadget do not participate in any matching and by Lemma~\ref{allmismatches} we have that all symbols $\mathtt{b}'$ from the gadget are mismatched among themselves. Therefore, we can assume that $l\leq \ell_4$. We need to consider two subcases.
			\begin{itemize}
				\item 
					$l$ is even. The symbols in $s$ contribute at least $S_1=(2\ell_4-l)/2$ to the Dyck score. We do the following modification to the algorithm. We match all symbols $\mathtt{b}$ to their counterparts $\mathtt{b}'$. We mismatch $l$ symbols in pairs among themselves. After the modification, the $Dyck$ contribution of symbols from $s$ is $S_2:=l/2$. We see that $S_2\leq S_1$.
				\item
					$l$ is odd. The symbols in $s$ contribute at least $S_1=(2\ell_4-l+1)/2$ to the Dyck score (at least one symbol is deleted because $2\ell_4-l$ is odd). We do the following modification to the algorithm. We match all symbols $\mathtt{b}$ to their counterparts $\mathtt{b}'$. We mismatch $l$ symbols in pairs among themselves except that we delete one symbol. After the modification, the $Dyck$ contribution of symbols from $s$ is $S_2:=(l+1)/2$. We see that $S_2\leq S_1$.
			\end{itemize}
		\end{itemize}
		Because symbols in between $\mathtt{b}$ and $\mathtt{b}'$ don't have their counterparts among themselves, the gadget contributes $E_l:=(|CG_\beta(t)|-2\ell_4)/2$ to the $Dyck$ cost. This quantity only depends on $n$ and $k$ and this can be verified from the construction of $S_G$.
	\end{proof}
	
	\begin{claim} \label{existmatches}
		In some optimal alignment, we can assume that there is some symbol $\mathtt{x}_\beta$ that is aligned to its counterpart $\mathtt{x}_\beta'$. Analogous statements can be proved about symbols $\mathtt{y}_\beta$, $\mathtt{y}_\beta'$, $\mathtt{x}_\alpha$, $\mathtt{x}_\alpha'$, $\mathtt{y}_\alpha$, $\mathtt{y}_\alpha'$, $\mathtt{x}_\gamma$, $\mathtt{x}_\gamma'$, $\mathtt{y}_\gamma$, $\mathtt{y}_\gamma'$.
	\end{claim}
	\begin{proof}
		Suppose that $\mathtt{x}_\beta$ is not aligned to any symbol $\mathtt{x}_\beta'$.
		We will modify the alignment so that all symbols $\mathtt{x}_\beta$ are aligned to $\mathtt{x}_\beta'$ coming from the first substring $f$ of $S_G$ consisting entirely of $\mathtt{x}_\beta'$.
		From the statement we have that all symbol $\mathtt{x}_\beta$ and all symbols from $f$ are mismatched or deleted. Therefore, by Lemma \ref{allmismatches} we get that all symbols $\mathtt{x}_\beta$ are mismatched among themselves and all symbols in $f$ are mismatched among themselves. Now we modify the alignment as follows to achieve our goal. We delete all symbols $b$ between $\mathtt{x}_\beta$ and $f$. Also, we delete all the symbols that the deleted $b$s were aligned to before the deletion. This increases $Dyck$ cost by at most $2\ell_4$. Then we match all symbols $\mathtt{x}_\beta$ to $\mathtt{x}_\beta'$ in pairs. This decreases $Dyck$ cost by $\ell_5$. As a result we decreases the $Dyck$ cost because $\ell_5\gg 2\ell_4$.
	\end{proof}
	
	\begin{claim} \label{neighborgadgets}
		In some optimal alignment, there is a symbol $\mathtt{x}_\beta'$ that is matched to a symbol $\mathtt{x}_\beta$ and
		there is a symbol $\mathtt{y}_\beta$ that is matched to a symbol $\mathtt{y}_\beta'$ so that
		the symbols $\mathtt{x}_\beta'$ and $\mathtt{y}_\beta$ come from the same gadget $CG_\beta(t)$.
		Analogous statements can be proved about symbols $\mathtt{x}_\alpha$, $\mathtt{x}_\alpha'$, $\mathtt{y}_\alpha$, $\mathtt{y}_\alpha'$, $\mathtt{x}_\gamma$, $\mathtt{x}_\gamma'$, $\mathtt{y}_\gamma$, $\mathtt{y}_\gamma'$.
	\end{claim}
	\begin{proof}
		By Claim \ref{existmatches}, $\mathtt{x}_\beta$ is aligned to some sequence consisting of $\mathtt{x}_\beta'$. Suppose that the sequence comes from gadget $CG_\beta(t_1)$ for some $t_1$.
		Also, by Claim \ref{existmatches}, $\mathtt{y}_\beta'$ is aligned to some sequence consisting of $\mathtt{y}_\beta$. Suppose that the sequence comes from gadget $CG_\beta(t_2)$ for some $t_2$.
		We want to prove that $t_1=t_2$. Suppose that this is not the case and $CG_\beta(t_1)$ comes to the left of $CG_\beta(t_2)$ (the order can't be reverse by the construction of $S_G$ and because the alignments can't cross).
		Suppose that there is some other gadget $CG_\beta(t_3)$ in between $CG_\beta(t_1)$ and $CG_\beta(t_2)$. Then we can verify that $CG_\beta(t_1)$ satisfy the conditions of Claim \ref{close} and we can assume that all symbols in $CG_\beta(t_3)$ are aligned with symbols in $CG_\beta(t_3)$. Therefore, we can remove gadget $CG_\beta(t_3)$ from $S_G$ (because it does not interact with symbols outside it) and this decreases the $Dyck$ cost by $E_{l}$. We do that until $CG_\beta(t_1)$ is to the left from $CG_\beta(t_2)$ and they are neighboring. Now we will change the alignment so that $\mathtt{x}_\beta$ is aligned to a symbol $\mathtt{x}_\beta'$ from $CG_\beta(t_2)$ and as a result we will decrease $Dyck$ cost. Now we can verify from the construction of $S_G$ that symbols $\mathtt{b}$, $\mathtt{y}_\beta$, $\mathtt{b}'$ from $CG_\beta(t_1)$ do not participate in matches. Also, symbols $\mathtt{b}$ and $\mathtt{x}_\beta'$ from $CG_\beta(t_2)$ do not participate in matches. By Lemma \ref{allmismatches}, we conclude that all these symbols have mismatches among themselves. 
		Let $g$ be the sequence between symbols $\mathtt{x}_\beta'$ and $\mathtt{y}_\beta$ in $CG_\beta(t_1)$.
		Now we modify the alignment so that the symbols in sequence $g$ have mismatches only among themselves and the symbols that were aligned to symbols in $g$ are deleted. This increases $Dyck$ cost by at most $10|g|\leq 100\ell_3=:S_1$. Some symbols $\mathtt{x}_\beta'$ are aligned to symbols outside $CG_\beta(t_1)$. We transfer these alignmets of symbols $\mathtt{x}_\beta'$ from $CG_\beta(t_1)$ to symbols $\mathtt{x}_\beta'$ in $CG_\beta(t_2)$ so that we only have mismatches among $\mathtt{x}_\beta'$ in $CG_\beta(t_1)$ and we don't change the $Dyck$ cost. Now we can align all symbols $\mathtt{b}$ in $CG_\beta(t_1)$ to $\mathtt{b}'$ in $CG_\beta(t_1)$ in pairs. This decreases the $Dyck$ cost by $\ell_4$. In the end, we decreased the $Dyck$ cost by $\ell_4-S_1>0$ and we proved what we wanted.
	\end{proof}
	
	\begin{claim} \label{exactone}
		In some optimal alignment, the only symbols $\mathtt{x}_\beta'$ that symbols $\mathtt{x}_\beta$ are matched to, come from the same gadget $CG_\beta(t)$, and
		the only symbols $\mathtt{y}_\beta$ that symbols $\mathtt{y}_\beta'$ are matched to, come from the same gadget $CG_\beta(t)$.
		In both cases it is the same gadget $CG_\beta(t)$.
		Analogous statements can be proved about symbols $\mathtt{x}_\alpha$, $\mathtt{x}_\alpha'$, $\mathtt{y}_\alpha$, $\mathtt{y}_\alpha'$, $\mathtt{x}_\gamma$, $\mathtt{x}_\gamma'$, $\mathtt{y}_\gamma$, $\mathtt{y}_\gamma'$.
	\end{claim}
	\begin{proof}
		Suppose that $\mathtt{x}_\beta$ is matched to symbols $\mathtt{x}_\beta'$ coming from two different gadgets $CG_\beta(t_1)$ and $CG_\beta(t_2)$. $CG_\beta(t_1)$ comes earlier in $S_G$ than $CG_\beta(t_2)$. Assume that there is no gadget $CG_\beta(t_3)$ in between $CG_\beta(t_1)$ and $CG_\beta(t_2)$ in sequence $S_G$. We can make this assumption because otherwise we can remove $CG_\beta(t_3)$ as in Claim \ref{close}. We can check that all symbols $x_\beta'$ and $y_\beta$ from $CG_\beta(t_3)$ do not participate in matches and thus we satisfy the requirements of Claim~\ref{close}. Now we can check that symbols $\mathtt{b}$, $\mathtt{y}_\beta$, $\mathtt{b}'$ in $CG_\beta(t_1)$ do not participate in matches. Also, symbols between $\mathtt{x}_\beta'$ and $\mathtt{y}_\beta$ in $CG_\beta(t_1)$ do not participate in matches. Also, symbols $\mathtt{b}'$ do not participate in matches. Therefore, by Lemma \ref{allmismatches} we conclude that all these symbols participate in mismatches only among themselves. Let $X_1$ denote sequence of $\mathtt{x}_\beta'$ from $CG_\beta(t_1)$ and $X_2$ denote sequence of $\mathtt{x}_\beta'$ from $CG_\beta(t_2)$. By Lemma \ref{rightleft}, we can assume that symbols $\mathtt{x}_\beta'$ from $X_1$ and $X_2$ that are matched to $\mathtt{x}_\beta$ form suffix in both sequences and the rest of symbols in both sequences are mismatched among themselves or deleted. The corresponding modifications increases the $Dyck$ cost by at most $\leq 4=:S_1$. Let $Z_1$ be the suffix of $X_1$ and $Z_2$ be the suffix of $X_2$. $|Z_1|+|Z_2|$ is less or equal to the total number of symbols $\mathtt{x}_\beta$ in $S_G$ by the construction of $S_G$. Suppose that $|Z_1|$ is even. We mismatch all symbols in $Z_1$ among themselves and match resulting unmatched $|Z_1|$ symbols $\mathtt{x}_\beta$ to $\mathtt{x}_\beta'$ from $X_2$. This does not change $Dyck$ cost. Suppose that $|Z_1|$ is odd, then there is a deletion among symbols in $X_1$ that are not in $Z_1$. We do mismatches among symbols $Z_1$ and the one deleted. We match the resulting unmatched $|Z_1|$ symbols $\mathtt{x}_\beta$ to $\mathtt{x}_\beta'$ from $X_2$. We can check that we can do this matching so that the $Dyck$ cost do not increase. Now we can match all symbols $\mathtt{b}$ from $CG_\beta(t_1)$ to their counterparts $\mathtt{b}'$ in $CG_\beta(t_1)$. This decreases the $Dyck$ cost by $S_2:=\ell_4$. In total, we decrease the $Dyck$ cost by $\geq S_2-S_1>0$.
	\end{proof}

	We note that in the proof of Claim \ref{exactone} we remove $3N-3$ cliques from the graph ($N$ is the number of $k$-cliques in the graph), each removal costing $E_l$. After all the removals, we arrive to a sequence of form as required in Claim \ref{claim:clique}. Thus, we set $E_C:=(3N-3)E_l+E_c$ and our proof for Claim~\ref{claim:main} is finished.

	We are now ready to show that the construction of $S_G$ from the graph $G$ proves Theorem~\ref{thm:dyck}.

\begin{reminder}{Theorem~\ref{thm:dyck}}
If Dyck edit distance on a sequence of length $n$ can be solved in $T(n)$ time, then $3k$-Clique on $n$ node graphs can be solved in $O\left(T \left(n^{k+O(1)} \right) \right)$ time, for any $k\geq 1$. 
Moreover, the reduction is combinatorial.
\end{reminder}

\begin{proof}
Given a graph $G$ on $n$ nodes we construct the sequence $S_G$ as described above.
The sequence can be constructed in $O(k^{O(1)}\cdot n^{k+O(1)})$ time, by enumerating all subsets of $k$ nodes and that it has length $O(k^{O(1)}\cdot n^{k+O(1)})$.
Thus, an algorithm for Dyck Edit Distance as in the statement returns $Dyck$ score of $S_G$ in $T(n^{k/3+O(1)})$ time (treating $k$ as a constant) and by Claim~\ref{claim:main} this score determines whether $G$ contains a $3k$-clique.
All the steps in our reduction are combinatorial.
\end{proof}